%% file: main.tex
\documentclass[11pt]{article}
\input{defs}

\usepackage{fullpage}

\usepackage[letterpaper,
left=1in,right=1in,top=1in,bottom=1in]{geometry}

\usepackage{amsmath}
\usepackage{amssymb}
\usepackage{amsthm}
\usepackage{graphicx}
\usepackage{float}
\usepackage{color}
\usepackage{pifont}
\usepackage[ruled, linesnumbered]{algorithm2e}
\usepackage{booktabs}
\usepackage{multirow}
\usepackage{threeparttable}
\usepackage{arydshln}
\usepackage{multirow}
\usepackage{threeparttable}
\usepackage[toc,page]{appendix}
\usepackage{bbm}
\usepackage{mathrsfs} % zotero refs need this sometimes
\usepackage{tablefootnote}
\usepackage{blkarray}
\usepackage{nicematrix}
\usepackage[utf8]{inputenc}
\usepackage{newunicodechar}
\usepackage{caption}
\usepackage[linktocpage]{hyperref}

\SetKwInput{KwData}{Input}
\SetKwInput{KwResult}{Output}

\hypersetup{
breaklinks=true,   % splits links across lines
colorlinks=true,   % displays links as colored text instead of blocks
citecolor=Blue,
linkcolor=BrickRed
}

\newcommand\QQ{\boldsymbol{\mathit{Q}}}

\def\calO{\mathcal{O}}
\def\calQ{\mathcal{Q}}
\def\calV{\mathcal{V}}
\def\rank{\text{rank}}

\def\im{\text{Im}}
\def\ker{\text{Ker}}

\def\xi{\xx^I}

\def\edges{E}

\def\nd{Nested Dissection}
\def\pcg{Preconditioned Conjugate Gradient}

\def\upsolver{\textsc{UpLapSolver}}
\def\downsolver{\textsc{DownLapSolver}}

\def\Lw{\LL}
\def\Ldown{\LL^{\text{down}}}
\def\Lup{\LL^{\text{up}}}

\def\fup{\ff^{\text{up}}}
\def\fdown{\ff^{\text{down}}}
\def\schur{\text{Sc}}
\def\diam{D}

\def\Pdown{\PPi^{\text{down}}}
\def\Pup{\PPi^{\text{up}}}
\def\aPdown{\widetilde{\PPi}^{\text{down}}}
\def\aPup{\widetilde{\PPi}^{\text{up}}}

\def\calU{\mathcal{U}}

\def\dim{d}
\def\p2{\partial_2}

\makeatletter
\renewcommand*\env@matrix[1][\arraystretch]{%
  \edef\arraystretch{#1}%
  \hskip -\arraycolsep
  \let\@ifnextchar\new@ifnextchar
  \array{*\c@MaxMatrixCols c}}
\makeatother

\begin{document}

\title{Efficient $1$-Laplacian Solvers for Well-Shaped Simplicial Complexes: Beyond Betti Numbers and Collapsing Sequences}

\author{
	Ming Ding\\ 
	\texttt{ming.ding@inf.ethz.ch}\\
	Department of Computer Science\\
	ETH Zurich
	\and
	Peng Zhang\thanks{This work was supported in part by NSF Grant CCF-2238682.}\\
	\texttt{pz149@cs.rutgers.edu}\\
	Department of Computer Science\\
	Rutgers University
}

\date{} % clear date

\clearpage\maketitle
\thispagestyle{empty}

\begin{abstract}
We present efficient algorithms for approximately solving systems of linear equations in $1$-Laplacians of well-shaped simplicial complexes up to high precision. $1$-Laplacians, or higher-dimensional Laplacians, generalize graph Laplacians to higher-dimensional simplicial complexes and play a key role in computational topology and topological data analysis. Previously, nearly-linear time approximate solvers were developed for simplicial complexes with known collapsing sequences and bounded Betti numbers, such as those triangulating a three-ball in $\mathbb{R}^3$ (Cohen, Fasy, Miller, Nayyeri, Peng, and Walkington [SODA'2014], Black, Maxwell, Nayyeri, and Winkelman [SODA'2022], Black and Nayyeri [ICALP'2022]). Furthermore, Nested Dissection provides quadratic time exact solvers for more general systems with nonzero structures representing well-shaped simplicial complexes embedded in $\mathbb{R}^3$.

We generalize the specialized solvers for $1$-Laplacians to simplicial complexes with additional geometric structures but \emph{without collapsing sequences and bounded Betti numbers}, and we improve the runtime of Nested Dissection. We focus on simplicial complexes that meet two conditions: (1) each individual simplex has a bounded aspect ratio, and (2) they can be divided into ``disjoint'' and balanced regions with well-shaped interiors and boundaries. Our solvers draw inspiration from the Incomplete Nested Dissection for stiffness matrices of well-shaped trusses (Kyng, Peng, Schwieterman, and Zhang [STOC'2018]).
\end{abstract}

\newpage
\pagenumbering{gobble}

\sloppy

{\small
	\tableofcontents
}

\newpage

\pagenumbering{arabic}

\input{intro}
\input{prelim}
\input{main_results}

\input{solveF}
\input{solve_schur}

\input{up_project}
\input{r_hollowing_new}

\input{another_precond_solver}
\input{multiple_balls}

\vspace{10pt}
\paragraph{Acknowledgements}
We thank Rasmus Kyng for the valuable discussions and the anonymous reviewers for their insightful comments.

\newpage
\bibliographystyle{alpha}
\bibliography{ref}

\appendix
\input{appendix}

\end{document}

%% file: defs.tex
\usepackage{blindtext}
\usepackage[usenames,dvipsnames,svgnames,table]{xcolor}
\usepackage{caption}
\usepackage{placeins}
\usepackage{graphicx}

\usepackage{fancybox}

\usepackage{amsmath,amsfonts,amsthm,amssymb,xcolor}
\usepackage{bm} % fixes boldsymbol

\usepackage{nicefrac}

\usepackage{float}

\newtheorem{theorem}{Theorem}[section]

\newtheorem{corollary}[theorem]{Corollary}
\newtheorem{lemma}[theorem]{Lemma}

\newtheorem{claim}[theorem]{Claim}

\newtheorem{fact}[theorem]{Fact}

\theoremstyle{definition}
\newtheorem{definition}[theorem]{Definition}

\newtheorem*{theorem*}{Theorem}
\newtheorem*{corollary*}{Corollary}
\newtheorem*{conjecture*}{Conjecture}
\newtheorem*{lemma*}{Lemma}
\newtheorem*{thm*}{Theorem}
\newtheorem*{prop*}{Proposition}
\newtheorem*{obs*}{Observation}
\newtheorem*{definition*}{Definition}

\newtheorem*{remark*}{Remark}
\newtheorem*{rec*}{Recommendation}

\newenvironment{fminipage}%
{\begin{Sbox}\begin{minipage}}%
		{\end{minipage}\end{Sbox}\fbox{\TheSbox}}

\def\pleq{\preccurlyeq}
\def\pgeq{\succcurlyeq}

\def\defeq{\stackrel{\mathrm{def}}{=}}

\def\dim#1{\mathrm{dim} (#1)}

\def\union{\cup}

\def\abs#1{\left|#1  \right|}

\def\norm#1{\left\| #1 \right\|}

\newcommand\calB{\mathcal{B}}
\newcommand\calC{\mathcal{C}}
\newcommand\calD{\mathcal{D}}

\newcommand\calG{\mathcal{G}}

\newcommand\calS{\mathcal{S}}

\newcommand\calT{\mathcal{T}}

\newcommand\PPi{\boldsymbol{\Pi}}

\newcommand\ddelta{\boldsymbol{\delta}}

\def\aa{\pmb{\mathit{a}}}
\newcommand\bb{\boldsymbol{\mathit{b}}}

\newcommand\ff{\boldsymbol{\mathit{f}}}

\newcommand\hh{\boldsymbol{\mathit{h}}}

\newcommand\uu{\boldsymbol{\mathit{u}}}
\newcommand\vv{\boldsymbol{\mathit{v}}}
\newcommand\ww{\boldsymbol{\mathit{w}}}
\newcommand\yy{\boldsymbol{\mathit{y}}}
\newcommand\zz{\boldsymbol{\mathit{z}}}
\newcommand\xx{\boldsymbol{\mathit{x}}}

\newcommand\veczero{\boldsymbol{0}}

\renewcommand\AA{\boldsymbol{\mathit{A}}}
\newcommand\BB{\boldsymbol{\mathit{B}}}

\newcommand\DD{\boldsymbol{\mathit{D}}}

\newcommand\HH{\boldsymbol{{H}}}
\newcommand\II{\boldsymbol{\mathit{I}}}
\newcommand\JJ{\boldsymbol{\mathit{J}}}

\newcommand\MM{\boldsymbol{\mathit{M}}}
\newcommand\LL{\boldsymbol{\mathit{L}}}
\newcommand\PP{\boldsymbol{\mathit{P}}}

\newcommand\TT{\boldsymbol{\mathit{T}}}
\newcommand\UU{\boldsymbol{\mathit{U}}}
\newcommand\WW{\boldsymbol{\mathit{W}}}
\newcommand\VV{\boldsymbol{\mathit{V}}}
\newcommand\XX{\boldsymbol{\mathit{X}}}

\newcommand\ZZ{\boldsymbol{\mathit{Z}}}

\newcommand\BBhat{\boldsymbol{\widehat{\mathit{B}}}}

\def\calC{\mathcal{C}}
\def\calD{\mathcal{D}}

\def\calG{\mathcal{G}}

\def\calK{\mathcal{K}}

\def\calS{\mathcal{S}}
\def\calX{\mathcal{X}}
\def\calY{\mathcal{Y}}

\newcommand\eps{\epsilon}

\newcommand\Otil{\widetilde{O}}

\newcommand\R{\mathbb{R}}

\DeclareMathOperator*{\diag}{diag}
\DeclareMathOperator*{\Span}{span}

\DeclareMathOperator*{\im}{im}

%% file: intro.tex
%!TEX root = main.tex

\clearpage

\def\bup{\bb^{\text{up}}}
\def\abup{\widetilde{\bb}^{\text{up}}}
\def\bdown{\bb^{\text{down}}}
\def\abdown{\widetilde{\bb}^{\text{down}}}
\def\xup{\xx^{\text{up}}}
\def\polylog{\text{polylog}}

\section{Introduction}
\label{sect:intro}

Combinatorial Laplacians generalize graph Laplacian matrices to higher dimensional simplicial complexes -- a collection of $0$-simplexes (vertices), $1$-simplexes (edges), $2$-simplexes (triangles), and their higher dimensional counterparts.
Simplicial complexes encode higher-order relations between data points in a metric space. 
By studying the topological properties of these complexes using Combinatorial Laplacians, one can capture higher-order features that go beyond connectivity and clustering.

Given an oriented $d$-dimensional simplicial complex $\calK$, for each $0 \le i \le d$, let $\calC_i$ be the vector space generated by the $i$-simplexes in $\calK$ with coefficients in $\mathbb{R}$.
We can define a sequence of boundary operators: 
\[
  \calC_d  \xrightarrow[]{\partial_{d}} \calC_{d-1} 
  \xrightarrow[]{\partial_{d-1}} \cdots 
  \xrightarrow[]{\partial_{2}} \calC_1 
  \xrightarrow[]{\partial_{1}} \calC_0,
\]
where each $\partial_i$ is a linear map that maps every $i$-simplex to a signed sum of its boundary $(i-1)$-faces. 
We define the $i$-Laplacian $\Lw_i: \calC_i \rightarrow \calC_i$ to be 
\begin{align}
  \Lw_i = \partial_{i+1} \partial_{i+1}^\top + \partial_{i}^\top \partial_{i}.
\label{eqn:def_Li}
\end{align}
In particular, $\partial_1$ is the vertex-edge incidence matrix, and $\LL_0$ is the graph Laplacian (following the convention, we define $\partial_0 = {\bf 0}$).
One can assign weights to each simplex in $\calK$ and define weighted Laplacians.

It is well-known that linear equations in graph Laplacians can be approximately solved in nearly-linear time in the number of nonzeros of the system \cite{ST14,KMP10,KMP11,KOSZ13,LS13,PS14,CKMPPRX14,KS16,KLPSS16,JS21}. 
These fast Laplacian solvers have led to significant developments in algorithm design for graph problems such as maximum flow \cite{madry2013navigating,madry2016computing,CKLPGS22}, minimum cost flow and lossy flow \cite{lee2014path,daitch2008faster}, and graph sparsification \cite{spielman2008graph}, known as ``the Laplacian Paradigm'' \cite{teng10}.

Inspired by the success of graph Laplacians,
Cohen, Fasy, Miller, Nayyeri, Peng, and Walkington \cite{CMFNPW14} initiated the study of fast approximate solvers for $1$-Laplacian linear equations.
They designed a nearly-linear time approximate solver for simplicial complexes \emph{with zero Betti numbers}\footnote{Informally, the $i$th Betti number is the number of $i$-dimensional holes on a topological surface. For example, the zeroth, first, and second Betti numbers represent the numbers of connected components, one-dimensional ``circular'' holes, and two-dimensional ``voids'' or ``cavities,'' respectively. 
} \emph{and known collapsing sequences}.
Later, Black, Maxwell, Nayyeri and Winkelman \cite{BMNW22}, and Black and Nayyeri \cite{BN22} generalized this algorithm to subcomplexes of such a complex \emph{with bounded first Betti numbers}\footnote{The solver has cubic dependence on the first Betti number.}.
One concrete example studied in these papers is convex simplicial complexes that piecewise linearly triangulate a convex ball in $\R^3$, 
for which a collapsing sequence exists and can be computed in linear time \cite{C67,C80}. 
However, deciding whether a simplicial complex has a collapsing sequence is NP-hard in general \cite{T16}; 
computing the Betti numbers is as hard as computing the ranks of general $\{0,1\}$ matrices \cite{EP14}.
In addition, $1$-Laplacian systems for general simplicial complexes embedded in $\mathbb{R}^4$ are as hard to solve as general sparse linear equations \cite{DKPZ22}, for which the best-known algorithms need super-quadratic time \cite{PV21,nie22}.
All the above motivates the following question:
\begin{quote}
    Can we efficiently solve $1$-Laplacian systems for other classes of \emph{structured} simplicial complexes, e.g., \emph{without known collapsing sequences and with arbitrary Betti numbers}?
\end{quote}

In addition to the specialized solvers for $1$-Laplacian systems mentioned above, \nd \ can solve $1$-Laplacian systems in quadratic time for simplicial complexes in $\R^3$ with additional geometric structures \cite{george73,LRT79,MT90} such as bounded aspect ratios\footnote{The aspect ratio of a geometric shape $S$ is the radius of the smallest ball containing $S$ divided by the radius of the largest ball contained in $S$.} of individual tetrahedrons. 
Furthermore, iterative methods such as Preconditioned Conjugate Gradient approximately solve $1$-Laplacian systems 
in time $\Otil(n \sqrt{\kappa})$\footnote{We use $\Otil(\cdot)$ to hide $\polylog$ factors on the number of simplexes, the condition number, and the inverse of the error parameter.}, where $n$ is the number of simplexes and $\kappa$ is the condition number of the coefficient matrix.

Inspired by solvers that leverage both geometric structures and spectral properties, we develop efficient $1$-Laplacian approximate solvers for well-shaped simplicial complexes embedded in $\R^3$  \emph{without known collapsing sequences and with arbitrary Betti numbers}.
Our solver adapts the Incomplete Nested Dissection algorithm, proposed by Kyng, Peng, Schwieterman, and Zhang \cite{KPSZ18} for solving linear equations in well-shaped 3-dimensional truss stiffness matrices.
These matrices represent another generalization of graph Laplacians; however, they differ quite from the $1$-Laplacians studied in this paper. A primary distinction is that the kernel of a truss stiffness matrix has an explicit and well-understood form, while computing a $1$-Laplacian's kernel is as hard as that for a general matrix.

\subsection{Our Results}

We say a simplex is \emph{stable} if it has $O(1)$ aspect ratio and $\Theta(1)$ weight. 
We focus on a \emph{pure} simplicial complex\footnote{A simplicial complex is \emph{pure} if every maximal simplex (i.e., a simplex that is not a proper subset of any other simplex in the complex) has the same dimension.
For example, a pure $3$-complex is a tetrahedron mesh that consists of tetrahedrons and their sub-simplexes.} $\calK$ embedded in $\R^3$.
We require $\calK$ admits a nice division parameterized by $r \in \mathbb{R}_{+}$, called \emph{$r$-hollowing}. 
The concept of $r$-hollowing was first introduced in \cite{KPSZ18} and can be viewed as a well-shaped $r$-division in $3$ dimensions. 
The structure $r$-division is frequently employed to accelerate algorithms by utilizing geometric structures, particularly for planar graphs \cite{federickson87,goodrich92,KMS13}. It facilitates the divide-and-conquer approach.
We adapt the concept of $r$-hollowing to suit our 1-Laplacian solvers.
Informally, our $r$-hollowing for a simplicial complex containing $n$ simplexes divides $\calK$ into $O(n/r)$ ``separated'' regions where each region has $O(r)$ simplexes and $O(r^{2/3})$ boundary simplexes. Only boundary simplexes can appear in multiple regions. 
Additionally, we mandate that each region's boundary triangulates a spherical shell in $\R^3$, exhibiting a ``hop'' diameter of $O(r^{1/3})$ and a ``hop'' shell width of at least $5$.
The formal definition of $r$-hollowing is given in Definition \ref{def:hollow}.
The bounded aspect ratio of each tetrahedron allows us to employ \nd \ for the interior simplexes within every region. The boundary shape requirement facilitates preconditioning the sub-system, derived from partial \nd, by the boundaries themselves and solving this sub-system using \pcg.

Below, we present our main results informally. Firstly, we assume that an $r$-hollowing of a pure $3$-complex is provided, which offers the broadest applicability of our algorithm. 
This assumption is justifiable when one can determine the construction of the simplicial complex; for instance, one can decide how to discretize a continuous topological space or how to triangulate a space given a set of points. 
Subsequently, we establish sufficient conditions for $3$-complexes that allow us to compute an $r$-hollowing in nearly-linear time.

\begin{theorem}[Informal statement]
Let $\calK$ be a pure $3$-complex embedded in $\R^3$ and composed of $n$ stable simplexes. 
Given an $r$-hollowing for $\calK$, for any $\eps > 0$, we can approximately solve a system in the $1$-Laplacian of $\calK$ within error $\eps$ in time 
$O \left( nr + n^{4/3}r^{5/18} \log(n/\eps) + n^2 r^{-2/3} \right)$. The runtime is $o(n^2)$ if $r = o(n)$ and $r = \omega(1)$.
In particular, when $r = \Theta(n^{3/5})$, the runtime is minimized (up to constant) and equals $O(n^{8/5} \log (n/\eps))$.
\label{thm:main_ball_informal}
\end{theorem}

Our runtime in Theorem \ref{thm:main_ball_informal} does not depend on the Betti numbers of $\calK$ and does not require collapsing sequences. 
When $r=o(n)$ and $r = \omega(1)$, the runtime is $o(n^2)$, asymptotically faster than \nd~\cite{MT90}. 
The solver in \cite{BN22} for a $1$-Laplacian system for the $\calK$ stated in Theorem \ref{thm:main_ball_informal} is $\Otil(\beta^3 m)$, where $m$ is the number of simplexes in $\calX \supset \calK$ with a known collapsing sequence and $\beta$ is the first Betti number of $\calK$. In the worst-case scenario, $m$ can be as large as $\Omega(n^2)$. 
But \cite{BN22} does not require that simplexes are stable and $\calK$ has a known $r$-hollowing.

Without assuming prior knowledge about $r$-hollowing, the following theorem presents a solver with the same runtime as Theorem \ref{thm:main_ball_informal} when the simplicial complex $\calK$ satisfies additional geometric restrictions: First, the convex hull of $\calK$ has $O(1)$ aspect ratio, and each tetrahedron of $\calK$ has $\Theta(1)$ volume. 
Second, all except one boundary component of $\calK$, which correspond to ``holes inside'' $\calK$, satisfy the following conditions: 
(1) every boundary component of $\calK$ has $1$-skeleton diameter $O(r^{1/3})$; 
(2) the total size of boundary components within any $\mathbb{X} \subset \R^3$ of volume $r$ is at most $O(r^{2/3})$, and 
the total size of boundary components of $\calK$ is $O(nr^{-1/3})$;
(3) the distance between any two boundary components of $\calK$ is greater than $5$. 
These geometric conditions allow us to find an $r$-hollowing of $\calK$ in nearly-linear time.
We have the following result with the algorithm in Theorem \ref{thm:main_ball_informal}.

\begin{theorem}[Informal statement]
Let $\calK$ be a pure $3$-complex embedded in $\R^3$ and composed of $n$ stable simplexes; assume $\calK$ satisfies the aforementioned additional geometric structures with parameter $r$. Then, for any $\eps > 0$, we can approximately solve a system in the $1$-Laplacian of $\calK$ within error $\eps$ in time $O \left( nr + n^{4/3}r^{5/18} \log(n/\eps) + n^2 r^{-2/3} \right)$.
In particular, when $r = \Theta(n^{3/5})$, the runtime is minimized (up to constant) and equals $O(n^{8/5} \log (n/\eps))$.
\label{thm:main_ball_rhollow_informal}
\end{theorem}

We then examine unions of pure $3$-complexes glued together by identifying certain subsets of simplexes on the boundary components (called \emph{exterior simplexes}) of 3-complex chunks.
Moreover, each $3$-complex chunk admits a $\Theta(n_i^{3/5})$-hollowing with $n_i$ being the number of simplexes in this chunk.
We remark that such a union of $3$-complexes, called $\calU$, \emph{may not be embeddable in $\R^3$}. 
So, the previously established methods from \cite{CMFNPW14,BMNW22,BN22} and \nd \  are unsuitable for this scenario. Building on our algorithm for Theorem \ref{thm:main_ball_informal}, we design an efficient algorithm for $\calU$ whose runtime depends sub-quadratically on the size of $\calU$ and polynomially on the number of chunks and the number of simplexes shared by more than one chunk.

\begin{theorem}[Informal statement]
Let $\calU$ be a union of pure $3$-complexes that are glued together by identifying certain subsets of their exterior simplexes. Each $3$-complex chunk is embedded in $\R^3$, contains $n_i$ stable simplexes, and has a known $\Theta(n_i^{3/5})$-hollowing.
For any $\eps > 0$, we can solve a system in the $1$-Laplacian of $\calU$ within error $\eps$ in time 
$\Otil \left( n^{8/5} + n^{3/10} k^2 + k^3 \right)$,
where $n$ is the number of simplexes in $\calU$ and $k$ is the number of exterior simplexes shared by more than one complex chunk.
  \label{thm:main_many_balls_informal}
\end{theorem}

When $h = \Otil(1)$ and $k = \Otil(n^{1/2})$, the solver in Theorem \ref{thm:main_many_balls_informal} has the same runtime as Theorem \ref{thm:main_ball_informal}.
When $h = o(n^{2/5}), k = o(n^{3/5})$, the runtime is $o(n^{2})$, asymptotically faster than \nd.

\subsection{Motivations and Applications} 

In the past decades, Combinatorial Laplacians have played a crucial role in the development of computational topology and topological data analysis in various domains, such as statistics \cite{JLYT11,osting2013statistical}, graphics and imaging \cite{ma20111,tong2003discrete}, brain networks \cite{lee2019harmonic}, deep learning \cite{bronstein2017geometric}, signal processing \cite{barbarossa2020topological}, and cryo-electron microscope \cite{ye2017cohomology}.
We recommend readers consult accessible surveys \cite{ghrist08,C09,EH10,lim20} for more information.

Combinatorial Laplacians have their roots in the study of discrete Hodge decomposition \cite{eckmann44}, which states that the kernel of the $i$-Laplacian $\LL_i$ is isomorphic to the $i$th homology group of the simplicial complex. 
Among the many applications of combinatorial Laplacians, a central problem is determining the Betti numbers -- the ranks of the homology groups -- which are important topological invariants. 
% uniquely characterize the topology of a shape.
Additionally, discrete Hodge decomposition allows for the extraction of meaningful information from data by decomposing them into three mutually orthogonal components: gradient (in the image of $\partial_i^\top$), curl (in the image of $\partial_{i+1}$), and harmonic (in the kernel of $\LL_i$) components.
For instance, the three components of edge flows in a graph capture the global trends, local circulations, and ``noise''.

The computation of both Betti numbers and discrete Hodge decomposition of higher-order flows can be achieved by solving systems of linear equations in Combinatorial Laplacians \cite{friedman1996computing,lim20}. 
The rank of a matrix $\LL_i$ can be determined by solving a logarithmic number of linear equation systems in $\LL_i$ \cite{BV21}. 
The discrete Hodge decomposition can be calculated by solving least square problems involving boundary operators or Combinatorial Laplacians, which in turn reduces to solving linear equations in these matrices.

Furthermore, an important question in numerical linear algebra concerns whether the nearly-linear time solvers for graph Laplacian linear equations can be generalized to larger classes of linear equations. 
Researchers have achieved success with elliptic finite element systems \cite{boman2008solving}, Connection Laplacians \cite{KLPSS16}, directed Laplacians \cite{cohen2017almost,cohen2018solving}, 
well-shaped truss stiffness matrices  \cite{DS07,ST08,KPSZ18}.
It would be intriguing to determine what structures of linear equations facilitate faster solvers.
Another theoretically compelling reason for developing efficient solvers for $1$-Laplacians stems from the ``equivalence'' of time complexity between solving $1$-Laplacian systems and general sparse systems of linear equations \cite{DKPZ22}. 
If one can solve all $1$-Laplacian systems in time $\Otil((\text{\# of simplexes})^{c})$ where $c \ge 1$ is a constant, then one can solve all general systems of linear equations in time $\Otil((\text{\# of nonzero coefficients})^c)$.

\paragraph*{Organization of the Remaining Paper}
In the paper, we have organized the content into several sections. 
In Section \ref{sect:preliminaries}, we present the necessary background knowledge in linear algebra and topology. 
Section \ref{sect:algo_overview} formally states our main theorems, while Section \ref{sect:alg_overview_new} provides an overview of our algorithm ideas.
The subsequent sections focus on the detailed proofs of various technique ingredients of the algorithm. 
% \peng{I removed the rest since we haven't defined these concepts such as down- and up-Laplacian.}

% The subsequent sections focus on the detailed proofs of various technique ingredients of the algorithm. 
% Section \ref{sect:down_solver} includes the proof of the down-Laplacian solver. In Section \ref{sec:up_solver_merge}, we provide the proof of the up-Laplacian solver.
% % , where the former uses the idea of \nd~and the latter uses \pcg. 
% Section \ref{sect:up_proj} gives an algorithm for approximating the up-projection operator.
% In Section \ref{sect:together}, we combine all the components and provide a proof of Theorem \ref{thm:main_ball_informal}.
% % Additionally, we include supplementary ingredients in the paper. 
% In Section \ref{sect:hollowing}, we present a linear time algorithm for $r$-hollowing, assuming additional geometric structures of the complex, which proves Theorem \ref{thm:main_ball_rhollow_informal}.
% Section \ref{sect:solver_bdry} introduces a faster up-Laplacian solver using a slightly different approach to $r$-hollowing, which may be of independent interest. 
% Finally, in Section \ref{sect:many_balls}, we extend the 1-Laplacian solver for a 3-complex embedded in $\mathbb{R}^3$ to a union of 3-complexes and prove Theorem \ref{thm:main_many_balls_informal}.

% For the missing proofs, please refer to the Appendix.

%% file: prelim.tex
%!TEX root = main.tex

\section{Preliminaries}
\label{sect:preliminaries}

\subsection{Background of Linear Algebra}
\label{sect:linear_algebra}

Given a vector $\xx\in\R^n$, 
for $1\le i \le n$, we let $\xx[i]$ be the $i$th entry of $\xx$;
for $1 \le i < j \le n$, let $\xx[i:j]$ be $(\xx[i], \xx[i+1], \ldots, \xx[j])^\top$. 
The Euclidean norm of $\xx$ is $\norm{\xx}_2 \defeq \sqrt{\sum_{i=1}^n \xx[i]^2}$. 
Given a matrix $\AA \in \R^{m\times n}$, for $1 \le i \le m, 1 \le j \le n$,
we let $\AA[i,j]$ be the $(i,j)$th entry of $\AA$;
for $S_1 \subseteq \{1,\ldots,m \}, S_2 \subseteq \{1,\ldots,n\}$,
let $\AA[S_1, S_2]$ be the submatrix with row indices in $S_1$ and column indices in $S_2$.
Furthermore, we let $\AA[S_1, :] = \AA[S_1, \{1,\ldots,n\}]$ and $\AA[:,S_2] = \AA[\{1,\ldots,m\}, S_2]$.
The operator norm of $\AA$ (induced by the Euclidean norm) is 
$\norm{\AA}_2 \defeq \max_{\vv \in \R^n} \frac{\norm{\AA \xx}_2}{\norm{\xx}_2}$.
The image of $\AA$ is the linear span of the columns of $\AA$, denoted by $\im(\AA)$, and the kernel of $\AA$ to be $\{\xx \in \mathbb{R}^n: \AA \xx = {\bf 0}\}$, denoted by $\ker(\AA)$.
A fundamental theorem of Linear Algebra states $\R^m=\im(\AA)\oplus \ker(\AA^\top)$.

\begin{fact}\footnote{All the facts in this section are well-known. For completeness, we include their proofs in Appendix \ref{sect:proof_linear_algebra}.}
For any matrix $\AA \in \mathbb{R}^{m \times n}$, 
$\im(\AA) = \im(\AA \AA^\top)$.
\label{fact:image}
\end{fact}

\paragraph*{Pseudo-inverse and Projection Matrix}
The pseudo-inverse of $\AA$ is defined to be a matrix $\AA^\dagger$ that satisfies all the following four criteria: (1) $\AA\AA^\dagger\AA=\AA$, (2) $\AA^\dagger\AA\AA^\dagger=\AA^\dagger$, (3) $(\AA\AA^\dagger)^\top=\AA\AA^\dagger$, (4) $(\AA^\dagger\AA)^\top=\AA^\dagger\AA$.
The orthogonal projection matrix onto $\im(\AA)$ is $\PPi_{\im(\AA)}=\AA(\AA^\top\AA)^\dagger\AA^\top$.

\paragraph*{Eigenvalues and Condition Numbers}
Given a square matrix $\AA\in\R^{n\times n}$, let $\lambda_{\max}(\AA)$ be the maximum eigenvalue of $\AA$ and $\lambda_{\min}(\AA)$ the minimum \textit{nonzero} eigenvalue of $\AA$. 
The condition number of $\AA$, denoted by $\kappa(\AA)$, is the ratio between $\lambda_{\max}(\AA)$ and $\lambda_{\min}(\AA)$.
A symmetric matrix $\AA$ is \emph{positive semi-definite (PSD)} if all eigenvalues of $\AA$ are non-negative. 
Let $\BB \in \R^{n \times n}$ be another square matrix.
We say $\AA \pgeq \BB$ if $\AA-\BB$ is PSD. The \emph{condition number of $\AA$ relative to $\BB$} is
\[
    \kappa(\AA,\BB)
    \defeq \min\left\{ \frac{\alpha}{\beta}: \beta \PPi_{\im(\AA)} \BB \PPi_{\im(\AA)} \pleq \AA\pleq  
    \alpha \BB
    \right\}.
\]
\begin{fact}
    Let $\AA, \BB \in \mathbb{R}^{n \times n}$ be symmetric matrices such that $\AA \pleq \BB$. 
    Then, for any $\VV \in \mathbb{R}^{m \times n}$, $\VV \AA \VV^\top \pleq \VV \BB \VV^\top$.
    \label{fact:psd}
\end{fact}

\paragraph*{Schur Complement}
Let $\AA\in\R^{n\times n}$, and 
let $F \cup C$ be a partition of $\{1,\ldots, n\}$. We write $\AA$ as a block matrix:
\begin{align}
\AA = \begin{pmatrix}
    \AA[F,F] & \AA[F,C]\\
    \AA[C,F] & \AA[C,C]
\end{pmatrix}.
\label{eqn:blockA}
\end{align}
We define the (generalized) Schur complement of $\AA$ onto $C$ to be
\[
    \schur[\AA]_C = \AA[C,C] - \AA[C,F] \AA[F,F]^{\dagger} \AA[F,C].
\]
The Schur complement appears in performing a block Gaussian elimination on matrix $\AA$ to eliminate the indices in $F$.

\begin{fact}\label{clm:matrix_identity}
Let $\AA$ be a PSD matrix defined in Equation \eqref{eqn:blockA}. 
Then, 
\begin{align*}
  \AA = \begin{pmatrix}
    \II & \\
    \AA[C,F] \AA[F, F]^{\dagger} & \II 
  \end{pmatrix} 
  \begin{pmatrix}
    \AA[F, F] & \\
    & \schur[\AA]_C
  \end{pmatrix}
  \begin{pmatrix}
    \II & \AA[F, F]^{\dagger} \AA[F, C] \\
    & \II
  \end{pmatrix}.
\end{align*}
\label{fact:identity}
\end{fact}

\begin{fact}
Let $\AA$ be a PSD matrix defined in Equation \eqref{eqn:blockA}. Let $\AA = \BB \BB^\top$, and we decompose $\BB = \begin{pmatrix}
\BB_F \\
\BB_C
\end{pmatrix}$ accordingly. Then, 
$\schur[\AA]_C = \BB_C \PPi_{\ker(\BB_F)} \BB_C^\top$,
where $\PPi_{\ker(\BB_F)}$ is the projection onto the kernel of $\BB_F$.
\label{fact:schur_project_kernel}
\end{fact}

\paragraph*{Solving Linear Equations} 
We will need Fact \ref{fact:solvers} for relations between different error notations for linear equations and Theorem \ref{thm:PCG} for \pcg.

\begin{fact}
Let $\AA, \ZZ \in \R^{n \times n} $ be two symmetric PSD matrices, and let $\PPi_{\im(\AA)}$ be the orthogonal projection onto $\im(\AA)$.
\begin{enumerate}
    \item If $(1-\eps) \AA^{\dagger} \pleq \ZZ \pleq (1+\eps) \AA^{\dagger} $, then $\norm{\AA \ZZ \bb - \bb}_2 \le \eps \sqrt{\kappa(\AA)} \norm{\bb}_2 $ for any $\bb \in \im(\AA)$.
    \item If $\norm{\AA \ZZ \bb - \bb}_2 \le \eps \norm{\bb}_2$ for any $\bb \in \im(\AA)$, then $(1-\eps) \AA^{\dagger} \pleq \PPi_{\im(\AA)} \ZZ \PPi_{\im(\AA)} \pleq (1+\eps) \AA^{\dagger} $.
\end{enumerate}
\label{fact:solvers}
\end{fact}

\begin{theorem}[Preconditioned Conjugate Gradient \cite{axe85}]
\label{thm:PCG}
    Let $\AA,\BB\in\R^{n\times n}$ be two symmetric PSD matrices, and let $\bb\in\R^n$. Each iteration of Preconditioned Conjugate Gradient multiplies one vector with $\AA$, solves one system of linear equations in $\BB$, and performs a constant number of vector operations. For any $\epsilon>0$, the algorithm outputs an $\xx$ satisfying $\norm{\AA\xx-\PPi_{\im(\AA)} \bb}_2\leq \epsilon\norm{ \PPi_{\im(\AA)} \bb}_2$ in $O(\sqrt{\kappa}\log(\kappa/\epsilon))$ such iterations, where $\PPi_{\im(\AA)}$ is the orthogonal projection matrix onto the image of $\AA$ and $\kappa = \kappa(\AA, \BB)$.
\end{theorem}

\subsection{Background of Topology}

\paragraph*{Simplex and Simplicial Complexes}

We consider a \emph{$\dim$-simplex} (or $\dim$-dimensional simplex) $\sigma$ as 
an ordered set of $\dim+1$ vertices, 
denoted by $\sigma = [v_0,\ldots, v_\dim]$. 
A \emph{face} of $\sigma$ is a simplex obtained by removing a subset of vertices from $\sigma$.
A \emph{simplicial complex} $\calK$ is a finite collection of simplexes such that 
(1) for every $\sigma\in\calK$ if $\tau\subset \sigma$ then $\tau\in \calK$, 
and (2) for every $\sigma_1, \sigma_2\in\calK$, $\sigma_1\cap\sigma_2$ is either empty or a face of both $\sigma_1, \sigma_2$. 
The \emph{dimension} of $\calK$ is the maximum dimension of any simplex in $\calK$.
A \emph{$\dim$-complex} is a $\dim$-dimensional simplicial complex. 
For $1 \le i \le \dim$, the \emph{$i$-skeleton} of a $\dim$-complex $\calK$ is the subcomplex consisting of all the simplexes 
of $\calK$ of dimensions at most $i$. In particular, the $1$-skeleton of $\calK$ is a graph.

A \emph{piecewise linear embedding} of a $3$-complex in $\R^3$ maps a $0$-simplex to a point, a $1$-simplex to a line segment, a $2$-simplex to a triangle, and a $3$-simplex to a tetrahedron. In addition, the interior of the images of simplices are disjoint and the boundary of each simplex is mapped to the appropriate simplices.
Such an embedding of a simplicial complex $\calK$ defines an \emph{underlying topological space} $\mathbb{K}$ -- the union of the images of all the simplexes of $\calK$.
We say $\calK$ is \emph{convex} if $\mathbb{K}$ is convex.
We say $\calK$ \emph{triangulates} a topological space $\mathbb{X}$ if $\mathbb{K}$ is homeomorphic to $\mathbb{X}$.
A simplex $\sigma$ of $\calK$ is a \emph{exterior} simplex if $\sigma$ is contained in the boundary of $\mathbb{K}$. 
A connected component of exterior simplexes is called a \emph{boundary component} of $\calK$.

The \emph{aspect ratio} of a set $S \subset \R^3$ is the radius of the smallest ball containing $S$ divided by the radius of the largest ball contained in $S$. The aspect ratio of $S$ is always greater than or equal to $1$.
We say a simplex $\sigma$  is \emph{stable} if it has $O(1)$ aspect ratio and $\Theta(1)$ weight. 
Miller and Thurston proved the following lemma. As a corollary, the numbers of the vertices, the edges, the triangles, and the tetrahedrons of a $3$-complex $\calK$ that is composed of stable tetrahedrons are all equal up to a constant factor.

\begin{lemma}[Lemma 4.1 of \cite{MT90}]
Let $\calK$ be a $3$-complex in $\mathbb{R}^3$ in which each tetrahedron has $O(1)$ aspect ratio. Then, each vertex of $\calK$ is contained in at most $O(1)$ tetrahedrons.
\label{lem:max_deg}
\end{lemma}

\paragraph*{Boundary Operators}
An \emph{$i$-chain} is a weighted sum of the oriented $i$-simplexes in $\calK$ with the coefficients in $\mathbb{R}$.
Let $\calC_i$ denote the $i$th chain space. 
The \textit{boundary operator} is a linear map $\partial_i: \calC_i \rightarrow \calC_{i-1}$ such that for an oriented $i$-simplex $\sigma = [v_0, v_1, \ldots, v_i]$,
\[
	\partial_i(\sigma)=\sum_{j=0}^{i}(-1)^j[v_0,\ldots,\hat{v}_j,\ldots, v_i],
\]
where $[v_0,\ldots,\hat{v}_j,\ldots, v_i]$ is the oriented $(i-1)$-simplex obtained by removing $v_j$ from $\sigma$.

The operator $\partial_i$ can be written as a matrix in $\abs{\calC_{i-1}} \times \abs{\calC_{i}}$ dimensions, 
where the $(r,l)$th entry of $\partial_i$ is $\pm 1$ if the $r$th $(i-1)$-simplex is a face of the $l$th $i$-simplex and $0$ otherwise. See Figure \ref{fig:partial_example} for an example.
\begin{figure}[!ht] 
    \centering 
    \includegraphics[width=0.9\textwidth]{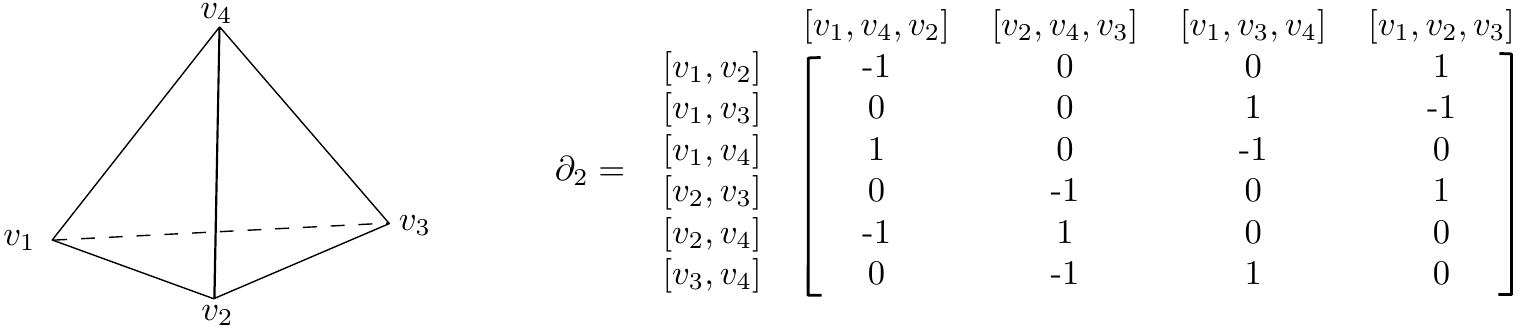} 
    \caption{An example of boundary operator. The left side is a 3-simplex (a tetrahedron) with vertices $v_1, v_2, v_3, v_4$. The right side is the corresponding second boundary operator $\partial_2$, where each column corresponds to an oriented $2$-simplex (a triangle) and each row corresponds to an oriented $1$-simplex (an edge).}
    \label{fig:partial_example} 
\end{figure}

An important property of boundary operators is 
$\partial_{i} \partial_{i+1} = {\bf 0}$, 
which implies $\im(\partial_{i+1}) \subseteq \ker(\partial_{i})$.
So, we can define the quotient space $H_{i} = \ker(\partial_{i}) \setminus \im(\partial_{i+1})$, called the \emph{$i$th homology space} of $\calK$. 
The rank of $H_i$ is called the $i$th Betti number of $\calK$. 
If the $i$th Betti number of $\calK$ is $0$, then 
$\im(\partial_i^\top) \oplus \im(\partial_{i+1}) = \mathbb{R}^{\abs{\calC_i}}$.
The first and second Betti numbers of a triangulation of a three-ball are both $0$.

\paragraph*{Hodge Decomposition and Combinatorial Laplacians}
Combinatorial Laplacians arise from the discrete Hodge decomposition.
\begin{theorem}[Hodge decomposition \cite{lim20}]
	Let $\AA\in\R^{m\times n}$ and $\BB\in\R^{n\times p}$ be matrices satisfying $\AA\BB=\mathbf{0}$. 
	Then, there is an orthogonal direct sum decomposition
	\[\R^n=\im(\AA^\top)\oplus\ker (\AA^\top\AA+\BB\BB^\top)\oplus\im(\BB).\]
\end{theorem}	
Since $\partial_{i} \partial_{i+1} = {\bf 0}$, it is valid to set $\AA=\partial_{i}$ and $\BB=\partial_{i+1}$. 
The matrix we get in the middle term is the \emph{combinatorial Laplacian}: $\LL_i\defeq\partial_{i}^\top \partial_{i}+\partial_{i+1} \partial_{i+1}^\top$.

The weighted combinatorial Laplacian generalizes combinatorial Laplacian.
For each $1 \le i \le d$, we assign each $i$-simplex of $\calK$ with a \emph{positive} weight, and let $\WW_i: \calC_i \rightarrow \calC_i$ be a diagonal matrix where $\WW_i[\sigma, \sigma]$ is the weighted of the $i$-simplex $\sigma$. 
Then the weighted \emph{$i$-Laplacian} of $\calK$ is a linear operator $\LL_i: \calC_i \rightarrow \calC_i$ defined as 
\[
\LL_{i} \defeq \partial_{i}^\top \WW_{i-1} \partial_{i} + \partial_{i+1} \WW_{i+1} \partial_{i+1}^\top.
\]
Note that Hodge decomposition also applies to weighted combinatorial Laplacian (by setting $\AA=\WW_{i-1}^{1/2}\partial_i$ and $\BB=\partial_{i+1}\WW_{i+1}^{1/2}$, we have $\AA\BB=\bf0$).
We call $\Ldown_i \defeq \partial_{i}^\top \WW_{i-1} \partial_{i}$ the $i$th \emph{down-Laplacian} operator and $\Lup_i \defeq \partial_{i+1} \WW_{i+1} \partial_{i+1}^\top$ the $i$th \emph{up-Laplacian} operator.
Sometimes, we use subscripts to specify the complex on which these operators are defined: $\partial_{i,\calK}, \WW_{i, \calK}, \Ldown_{i,\calK}, \Lup_{i, \calK}$.

\paragraph*{$r$-Hollowings}

Let $\calK$ be a pure $3$-complex with $n$ simplexes.
A set of triangles $\hat{\triangle}_1, \ldots, \hat{\triangle}_{k}$ form a \emph{triangle path} of length $k-1$ if for any $1 \le i \le k-1$, $\hat{\triangle}_i$ and $\hat{\triangle}_{i+1}$ share an edge.
The \emph{triangle distance} between two triangles $\triangle_1$ and $\triangle_2$ is the shortest triangle path length between $\triangle_1$ and $\triangle_2$.
The \emph{triangle diameter} of $\calK$ is the longest triangle distance between any two triangles.
A \emph{spherical shell} is $\{\xx \in \R^3: 
R_1 \le \norm{\xx}_2 \le R_2\}$ where $R_1 < R_2$.
If $\calK$ triangulates a spherical shell, we define the \emph{shell width} to be the shortest triangle distance between any two triangles where one is on the outer sphere and one is on the inner sphere.

\begin{definition}[$r$-hollowing]
Let $\calK$ be a $3$-complex with $n$ simplexes, and let $r = o(n)$ be a positive number.
We divide $\calK$ into $O(n/r)$ \emph{regions} each of $O(r)$ simplexes, $O(r^{2/3})$ \emph{boundary simplexes}, and $O(r^{2/3})$ exterior simplexes. Non-boundary simplexes are called \emph{interior simplexes}\footnote{We would like to emphasize that ``exterior simplex'' is defined for any $3$-complex, while ``boundary simplex'' and ``interior simplex'' are defined for $r$-hollowing. Although boundary and interior simplexes are mutually exclusive, an exterior simplex can be either a boundary or an interior simplex for a region of an $r$-hollowing.}. 
Interior simplexes from different regions do not share any subsimplexes.
The boundary of each region triangulates a spherical shell in $\R^3$ and has triangle diameter $O(r^{1/3})$ and shell width at least $5$.
The union of all boundary simplexes of each region is referred to as an \emph{$r$-hollowing} of $\calK$.
\label{def:hollow}
\end{definition}

Figure \ref{fig:r-hollowing-def} illustrates an example of $r$-hollowing by showing a cross-section of a $3$-complex. 
The left figure presents a cross-section of a $3$-complex $\calK$ with two holes inside, depicted as two empty discs. The outlines in black represent the ``exterior simplexes'' of $\calK$.
On the right-hand side of Figure \ref{fig:r-hollowing-def}, the gray area represents an $r$-hollowing of $\calK$. 
It divides $\calK$ into $4$ balanced regions $\{(1), (2), (3), (4)\}$, each indicated by the area inside the red, blue, green, or orange outline, respectively.
Region $(1)$ includes the smaller hole, while the larger hole ``intersects'' regions $(3)$ and $(4)$. 
The unshaded area inside the squares corresponds to ``interior simplexes,'' which encompasses the exterior simplexes of the smaller hole and the left half of the larger hole. In contrast, the gray area represents ``boundary simplexes,'' including the exterior simplexes of the right half of the larger hole. The shell width is indicated by arrows.

\begin{figure}[!ht] 
    \centering 
    \includegraphics[width=0.8\textwidth]{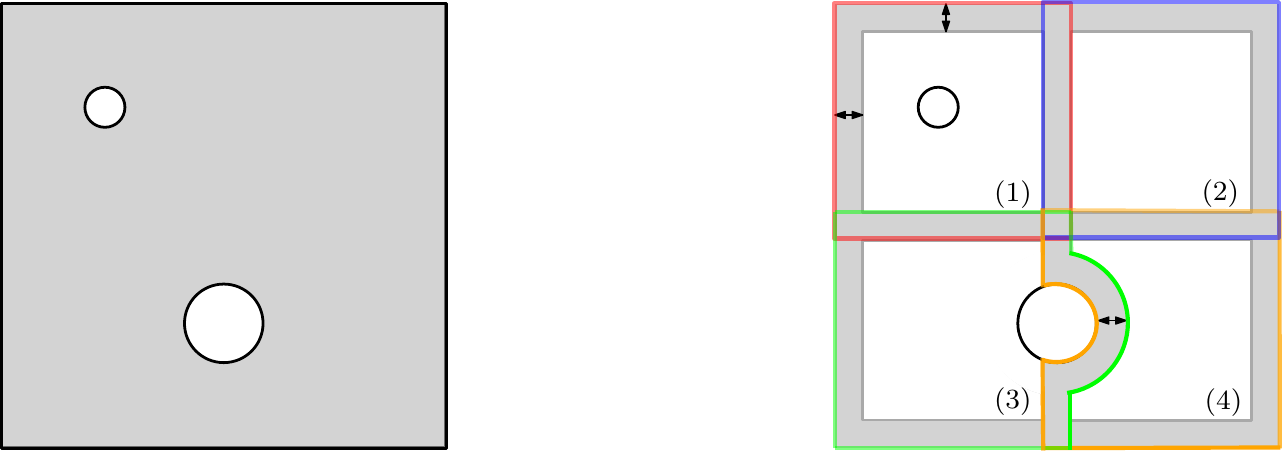} 
    \caption{A cross-sectional illustration of an $r$-hollowing.}
    \label{fig:r-hollowing-def} 
\end{figure}

In addition, this paper examines sufficient conditions for $3$-complexes that enable us to compute an $r$-hollowing in nearly-linear time. 
Specifically, we consider a pure $3$-complex $\calK$ embedded in $\R^3$ with $n$ stable simplexes each of volume $\Theta(1)$ possessing the following additional geometric structures:

\begin{definition}[Well-shaped boundary structure]
We say a $3$-complex $\calK$ has a \emph{well-shaped boundary structure} if:
\begin{enumerate}
\item \label{asp:1} The convex hull of $\calK$ has aspect ratio $O(1)$.
\item \label{asp:2} All except one boundary component have $1$-skeleton diameter $O(r^{1/3})$.
\item \label{asp:3} The total number of exterior simplexes of $\calK$ within any $\mathbb{X} \subset \R^3$ of volume $r$ is $O(r^{2/3})$; the total number of exterior simplexes of $\calK$ is $O(nr^{-1/3})$.
\item \label{asp:4} The triangle distance between any two boundary components of $\calK$ is greater than $5$. 
\end{enumerate}
\label{def:hole_requirement}
\end{definition}

It is worth noting that fulfilling the aforementioned assumptions is not excessively challenging. 
On one end of the spectrum, there are scenarios where $\calK$ contains at most $O(n/r)$ 2-dimensional holes (corresponding to the boundaries of $1$-skeleton diameters $O(r^{1/3})$), each with an interior volume of $O(r)$. 
On the other end, there are instances where $\calK$ encompasses $O(nr^{-1/3})$ uniformly distributed small holes, each with a constant interior volume. 
Moreover, it is likely that all scenarios lying between these extremes would also meet these assumptions.

%% file: main_results.tex
%!TEX root = main.tex

\def\fup{\ff^{\text{up}}}
\def\fdown{\ff^{\text{down}}}
\def\xup{\xx^{\text{up}}}
\def\xdown{\xx^{\text{down}}}
\def\upFsolve{\textsc{UpLapFSolver}}
\def\schursolve{\textsc{SchurSolver}}
\def\downSolver{\textsc{DownLapSolver}}
\def\lapSolver{\textsc{OneLapSolver}}
\def\union1solve{\textsc{UnionOneLapSolver}}

\def\upFa{\widetilde{\Lup_1}[F,F]}
\def\LupS{\Lup_{1,\calS}}
\def\p2S{\partial_{2,\calS}}
\def\w2S{\WW_{2,\calS}}
\def\LupOneS{\widehat{\Lup_{1,\calS}}}
\def\SchurLupOneS{\schur[\widehat{\Lup_{1,\calS}}]_{\edges_2}}

\def\aPdown{\widetilde{\PPi}^{\text{down}}}
\def\aPup{\widetilde{\PPi}^{\text{up}}}
\def\Zup{\ZZ^{\text{up}}}
\def\Zdown{\ZZ^{\text{down}}}
\def\axup{\widetilde{\xx}^{\text{up}}}
\def\axdown{\widetilde{\xx}^{\text{down}}}

\def\dims{\text{dim}}

\def\downsolver{\textsc{DownSolver}}

\def\PkerF{\PPi_{\ker(\partial_2^\top[F,:])}}

\section{Main Theorems}
\label{sect:algo_overview}

We formally state our main results as follows.

\begin{theorem}
  Let $\calK$ be a pure $3$-complex embedded in $\R^3$ consisting of $n$ stable simplexes and with a known $r$-hollowing. 
  Let $\LL_1$ be the $1$-Laplacian operator of $\calK$, and 
  let $\PPi_1$ be the orthogonal projection matrix onto the image of $\LL_1$.
  For any vector $\bb$ and $\eps > 0$, we can find a solution $\tilde{\xx}$ such that 
  $\norm{\LL_1 \tilde{\xx} - \PPi_1 \bb}_2
  \le \eps \norm{\PPi_1 \bb}_2$
  in time $O \left( nr + n^{4/3}r^{5/18} \log(n/\eps) + n^2 r^{-2/3} \right)$.
  \label{thm:main_ball}
\end{theorem}

We will overview our algorithm for Theorem \ref{thm:main_ball} in Section \ref{sect:alg_overview_new} and provide detailed proofs in Section \ref{sect:down_solver}, \ref{sec:up_solver_merge}, 
\ref{sect:up_proj}, and \ref{sect:together}.

\begin{theorem}
  Let $\calK$ be a pure $3$-complex embedded in $\R^3$ consisting of $n$ stable simplexes each of volume $\Theta(1)$.
  Suppose $\calK$ has a well-shaped boundary structure with parameter $r = o(n)$ defined as in Definition \ref{def:hole_requirement}. 
  Let $\LL_1$ be the $1$-Laplacian operator of $\calK$, and 
  let $\PPi_1$ be the orthogonal projection matrix onto the image of $\LL_1$.
  For any vector $\bb$ and $\eps > 0$, we can find a solution $\tilde{\xx}$ such that 
  $\norm{\LL_1 \tilde{\xx} - \PPi_1 \bb}_2
  \le \eps \norm{\PPi_1 \bb}_2$
  in time $O \left( nr + n^{4/3}r^{5/18} \log(n/\eps) + n^2 r^{-2/3} \right)$.
  \label{thm:main_ball_rhollowing}
\end{theorem}

The known $r$-hollowing assumption is replaced with the well-shaped boundary structure assumption in Theorem \ref{thm:main_ball_rhollowing}, and a nearly-linear time algorithm for finding $r$-hollowing is presented in Section \ref{sect:hollowing}.
It is worth mentioning that the additional geometric structures (Definition \ref{def:hole_requirement}) are introduced to ensure the feasibility of finding an $r$-hollowing in nearly-linear time. However, the algorithm for solving the system of linear equations remains the same.

\begin{theorem}
Let $\calU$ be a union of pure $3$-complexes glued together by identifying certain subsets of their exterior simplexes. Each $3$-complex chunk is embedded in $\R^3$ and comprises $n_i$ stable simplexes, and has a known $\Theta (n_i^{3/5})$-hollowing.
Let $\LL_1$ be the $1$-Laplacian operator of $\calU$, and let $\PPi_1$ be the orthogonal projection matrix onto the image of $\LL_1$.
For any vector $\bb$ and $\eps > 0$, we can find a solution $\tilde{\xx}$ such that 
$\norm{\LL_1 \tilde{\xx} - \PPi_1 \bb}_2
\le \eps \norm{\PPi_1 \bb}_2$
in time $\Otil(n^{8/5} + n^{3/10} k^2 + k^3)$,
where $n$ is the number of simplexes in $\calU$ and $k$ is the number of exterior simplexes shared by more than one chunk.
  \label{thm:main}
\end{theorem}

We will defer our proof of Theorem \ref{thm:main} to Section \ref{sect:many_balls}.

\section{Algorithm Overview}
\label{sect:alg_overview_new}

Cohen, Fasy, Miller, Nayyeri, Peng, and Walkington \cite{CMFNPW14} observed that 
\[
\LL_1^{\dagger} = \left( \Ldown_1 \right)^{\dagger} + \left( \Lup_1 \right)^{\dagger},
\]
where $\Ldown_1 = \partial_1^\top \WW_0 \partial_1$ is the down-Laplacian and $\Lup_1 = \partial_2 \WW_2 \partial_2^\top$ is the up-Laplacian.
The orthogonal projection matrices onto $\im(\partial_1^\top)$ and $\im(\partial_2)$ are:
\begin{align*}
\Pdown_1 \defeq \partial_1^\top(\partial_1\partial_1^\top)^\dagger\partial_1, ~ 
\Pup_1 \defeq \partial_2(\partial_2^\top\partial_2)^\dagger\partial_2^\top.
\end{align*}

\begin{lemma}[Lemma 4.1 of \cite{CMFNPW14}]
Let $\bb$ be a vector. Consider the systems of linear equations: $\LL_1 \xx = \PPi_{1} \bb, 
\Lup_1 \xx^{\text{up}} = \Pup_1 \bb, \Ldown_1 \xx^{\text{down}} = \Pdown_1 \bb$. Then, 
$\xx = \Pup_1 \xx^{\text{up}} + \Pdown_1 \xx^{\text{down}}$.
\label{lem:cohen_1Lap}
\end{lemma}
Lemma \ref{lem:cohen_1Lap} implies that four operators are needed to approximate $\LL_1^{\dagger}$:
(1) an approximate projection operator $\aPdown_1 \approx \Pdown_1$, 
(2) an approximate projection operator 
$\aPup_1 \approx \Pup_1$,
(3) a down-Laplacian solver $\ZZ^{\text{down}}_1$ such that $\Ldown_1 \ZZ^{\text{down}}_1 \bb \approx \bb$ for any $\bb \in \im(\Lup_1)$, and (4) an up-Laplacian solver $\ZZ^{\text{up}}_1$ such that $\Lup_1 \ZZ^{\text{up}}_1 \bb \approx \bb $ for any $\bb \in \im(\Lup_1)$.

We will apply the same approximate orthogonal projection $\aPdown_1$ given in \cite{CMFNPW14}, which does not depend on Betti numbers.
Our solver for the first down-Laplacian is a slight modification of the one in \cite{CMFNPW14} to incorporate simplex weights. We state the two lemmas but defer their proofs to Section \ref{sect:down_solver}.

\begin{lemma}[Down-projection operator, Lemma 3.2 of \cite{CMFNPW14}]
Let $\calK$ be a $3$-complex with $n$ simplexes. For any $\eps > 0$, there exists a linear operator $\aPdown_1$ such that 
\begin{align*}
    & (1-\eps)\Pdown_1 \pleq \aPdown_1(\eps) \pleq \Pdown_1.
\end{align*}
\label{lem:cohen_project}
\end{lemma}
\begin{lemma}[Down-Laplacian solver]
Let $\calK$ be a weighted simplicial complex, and let $\bb \in \im(\Ldown_1)$.
There exists an operator $\Zdown_1$ such that 
$\Ldown_1 \Zdown_1 \bb = \bb$. In addition, we can compute $\Zdown_1 \bb$ in linear time.
\label{lem:down_lap_solver}
\end{lemma}

\subsection{Solver for Up-Laplacian}

One of our primary technical contributions is the development of an efficient solver for the up-Laplacian system, stated in Lemma \ref{lem:up_solver}.
We will describe the key idea behind our solver in this section.

\begin{lemma}[Up-Laplacian solver]
  Let $\calK$ be a pure $3$-complex embedded in $\R^3$ and composed of $n$ stable simplexes. Suppose we are given an $r$-hollowing for $\calK$ where $r = o(n)$.
  Then for any $\eps > 0$, there exists an operator $\Zup_1$ such that
  \[
  \forall \bb \in \im(\Lup_1), ~ \norm{\Lup_1 \Zup_1 \bb - \bb}_2 \le \eps \norm{\bb}_2.
  \]
  In addition, $\Zup_1 \bb$ can be computed in time 
  $O \left( nr + n^{4/3}r^{5/18} \log(n/\eps) + n^2 r^{-2/3} \right)$.
  \label{lem:up_solver}
\end{lemma}

We remark that Lemma \ref{lem:up_solver} can be improved to $\Otil(n^{3/2})$ by using a slightly different $r$-hollowing (proved in Section \ref{sect:solver_bdry}), which might be of independent interest. Since the bottleneck of our solver for $1$-Laplacians is from the projection for up $1$-Laplaicans, we use the same $r$-hollowing here.

The given $r$-hollowing suggests a partition of the edges in $\calK$ into $F \cup C$. We will explain the concrete partition shortly. 
We have the following matrix identity:
\begin{align*}
  \Lup_1 = \begin{pmatrix}
    \II & \\
    \Lup_1[C, F] \Lup_1[F, F]^{\dagger} & \II 
  \end{pmatrix} 
  \begin{pmatrix}
    \Lup_1[F, F] & \\
    & \schur[\Lup_1]_C
  \end{pmatrix}
  \begin{pmatrix}
    \II & \Lup_1[F, F]^{\dagger} \Lup_1[F, C] \\
    & \II
  \end{pmatrix},
  \label{eqn:matrix_identity2}
\end{align*}
where 
\[
\schur[\Lup_1]_C = \Lup_1[C, C] - \Lup_1[C, F] \Lup_1[F, F]^{\dagger} \Lup_1[F, C].
\]
The following Lemma \ref{lem:combine_solvers} reduces  (approximately) solving a system in $\Lup_1$ to (approximately) solving two systems in $\Lup_1[F,F]$ and one system in $\schur[\Lup_1]_C$, which is proved in Appendix \ref{sect:proof_algo_overview}.
It is worth noting that Lemma \ref{lem:combine_solvers} holds if we replace $\Lup_1$ with an arbitrary symmetric PSD matrix, and we will apply it or its variants for different PSD matrices in our solvers.
To avoid introducing additional notations, we state the lemma below in terms of $\Lup_1$.

\begin{lemma}
Suppose we have two operators (1) $\upFsolve(\cdot)$ such that given any $\bb \in \im(\Lup_1[F,F])$, $\upFsolve(\bb)$ returns a vector $\xx$ satisfying $\Lup_1[F,F] \xx = \bb$, 
and (2) $\schursolve(\cdot, \cdot)$ such that for any $\hh \in \im(\schur[\Lup_1]_C)$ and $\delta > 0$, $\schursolve(\hh, \delta)$ returns $\tilde{\xx}$ satisfying 
$\norm{\schur[\Lup_1]_C \tilde{\xx} - \hh}_2 \le \delta \norm{\hh}_2.$
Given any $\bb = \begin{pmatrix}
    \bb_F \\
    \bb_C
\end{pmatrix} \in \im(\Lup_1)$ and any $\eps > 0$, 
let
\begin{align}
\begin{split}
& \hh = \bb_C - \Lup_1[C,F] \cdot \upFsolve (\bb_F), \\
& \tilde{\xx}_C = \schursolve(\hh, \delta), \\
& \tilde{\xx}_F = \upFsolve \left( \bb_F - \Lup_1[F,C] \tilde{\xx}_C \right),
\end{split}
\label{eqn:xtilde}
\end{align}
where $\delta \le \frac{\eps}{1 + \norm{\Lup_1[C,F] \Lup_1[F,F]^{\dagger}}_2}$. 
Then, 
\[
\norm{\Lup_1 \tilde{\xx} - \bb}_2 \le \eps  \norm{\bb}_2,
\]
where $\tilde{\xx} = \begin{pmatrix}
    \tilde{\xx}_F \\
    \tilde{\xx}_C
\end{pmatrix}$.
Let $m_F = \abs{F}$ and $m_C = \abs{C}$, and let $\upFsolve$ have runtime $t_1(m_F)$ and $\schursolve$ have runtime $t_2(m_C)$. Then, we can compute $\tilde{\xx}$ in time $O(t_1(m_F) + t_2(m_C) + m_F + m_C)$. 
\label{lem:combine_solvers}
\end{lemma}

\subsubsection{Partitioning the Edges}

As suggested by Lemma \ref{lem:combine_solvers},
we want to partition the edges of $\calK$ into $F \cup C$ so that both systems in $\Lup_1[F,F]$ and the Schur complement $\schur[\Lup_1]_C$ can be efficiently solved. 
The given $O(n^{3/5})$-hollowing divides $\calK$ into \emph{``disjoint'' and balanced} regions with \emph{small} boundary.
Let $F$ be the set of the interior edges of the regions and $C$ be the set of the boundary edges.

We first show the interiors of different regions are ``disjoint'' in the sense that 
$\Lup_1[F,F]$ is a block diagonal matrix where each diagonal block corresponds to the interior of a region.
We can write $\Lup_1$ as the sum of rank-$1$ matrices that each corresponds to a triangle in $\calK$:
\begin{equation*}
    \Lup_1 = \partial_2 \WW_2 \partial_2^\top
= \sum_{\sigma: \text{triangle in } \calK} 
\WW_2[\sigma, \sigma] \cdot \partial_2[:, \sigma]  \partial_2[:, \sigma]^\top.
\end{equation*}
For any two edges $e_1, e_2$, $\Lup_1[e_1, e_2] = 0$ if and only if no triangle in $\calK$ contains both $e_1, e_2$.
By our definition of $r$-hollowing in Definition \ref{def:hollow}, for different regions $R_1, R_2$ of $\calK$ w.r.t. an $r$-hollowing, no triangle contains both an edge from $R_1$ and an edge from $R_2$.

In addition, the following lemma shows that the boundaries of the regions well approximate the Schur complement onto the boundaries.
The proof is in Section \ref{sect:condNum}.

\begin{lemma}[Spectral bounds for $r$-hollowing]
Let $\calK$ be a pure $3$-complex embedded in $\R^3$ composed of stable simplexes.  Let $\calT$ be an $r$-hollowing of $\calK$, and let $C$ be the edges of $\calT$. Then, 
$\Lup_{1,\calT} \pleq \schur[\Lup_1]_C \pleq O(r) \Lup_{1,\calT}.$
  \label{lem:eigS}
\end{lemma}

\subsubsection{Proof of Lemma \ref{lem:up_solver} for Up-Laplacian Solver}

Algorithm \ref{alg:solver} sketches a pseudo-code for our up-Laplacian solver.

\begin{algorithm}
    \caption{\upsolver($\calK, \calT, \bb, \eps$)}
    \label{alg:solver}
    \KwData{A pure $3$-complex $\calK$ of $n$ stable simplexes with up-Laplacian $\Lup_1$, an $r$-hollowing $\calT$, a vector $\bb \in \im(\Lup_1)$, an error parameter $\eps > 0$}
    \KwResult{An approximate solution $\tilde{\xx}$ such that $\norm{\Lup_1\tilde{\xx} -\bb}_2\leq \eps\norm{\bb}_2$}

    $F \leftarrow$ the interior edges of regions of $\calK$ w.r.t. $\calT$;   ~~~
    $C \leftarrow$ the boundary edges of regions.

    \upFsolve($\cdot$)$\leftarrow$ a solver by \nd \ that satisfies the requirement in Lemma \ref{lem:combine_solvers}. \label{lin:solver:F}

    \schursolve($\cdot, \cdot$)$\leftarrow$ a solver by \pcg \ with the preconditioner being the up-Laplacian of $\calT$ that satisfies the requirement in Lemma \ref{lem:combine_solvers}. \label{lin:solver:schur}

    $\tilde{\xx} \leftarrow$ computed by Equation \eqref{eqn:xtilde}

    \Return{solution $\tilde{\xx}$}
\end{algorithm}

By Lemma \ref{lem:combine_solvers}, the $\tilde{\xx}$ returned by Algorithm \ref{alg:solver} satisfies $\norm{\Lup_1 \tilde{\xx} - \bb}_2 \le \eps \norm{\bb}_2$. To bound the runtime of Algorithm \ref{alg:solver}, we need the following lemmas for lines 
\ref{lin:solver:F} and \ref{lin:solver:schur}.

\begin{lemma}[Solver for the ``$F$'' part]
Let $\calK$ be a pure $3$-complex embedded in $\R^3$ and composed of $n$ stable simplexes. Let $\calT$ be an $r$-hollowing of $\calK$, and let $F$ be the set of interior edges in each region of $\calK$ w.r.t. $\calT$.
Then, with a pre-processing time $O(n r)$, there exists a solver $\upFsolve(\cdot)$ such that given any $\bb_F \in im(\Lup_1[F,F])$, $\upFsolve(\bb_F)$ returns an $\xx_F$ such that
$\Lup_1[F,F]\xx_F = \bb_F$ in time $O(n r^{1/3} )$.
\label{lem:solve_interior}
\end{lemma}

By our choice of $F$, the matrix $\Lup_1[F,F]$ can be written as a block diagonal matrix where each block corresponds to a region of $\calK$ w.r.t. the $r$-hollowing $\calT$. 
Since each region is a $3$-complex in which every tetrahedron has an aspect ratio $O(1)$, we can construct the solver \upFsolve \ by \nd \ \cite{MT90}. However, since each row or column of $\Lup_1[F,F]$ corresponds to an edge in $\calK$, we need to turn the good \emph{vertex separators} in \cite{MT90} into good \emph{edge separators} for regions of $\calK$. We prove Lemma \ref{lem:solve_interior} in Section \ref{sect:solver_int}.

\begin{lemma}[Solver for the Schur complement]
Let $\calK$ be a pure $3$-complex embedded in $\R^3$ and composed of $n$ stable simplexes. 
Let $\calT$ be an $r$-hollowing of $\calK$, and let $C$ be the set of boundary edges of each region of $\calK$ w.r.t. $\calT$.
Then, with a pre-processing time $O(nr + n^2 r^{-2/3})$
there exists a solver $\schursolve(\cdot,\cdot)$ such that for any $\hh \in \im(\schur[\Lup_1]_C)$ and $\delta > 0$, $\schursolve(\hh, \delta)$ returns an $\tilde{\xx}_C$ such that $\norm{\schur[\Lup_1]_C \tilde{\xx}_C - \hh}_2 \le \delta \norm{\hh}_2$ in time 
$\Otil \left( nr^{5/6} + n^{4/3}r^{5/18}
\right)$.
\label{lem:solve_schur}
\end{lemma}

Our solver \schursolve \ is based on the \pcg \ (PCG) with the preconditioner $\Lup_{1,\calT}$, the up-Laplacian operator of $\calT$.
By Theorem \ref{thm:PCG} and Lemma \ref{lem:eigS}, the number of PCG iterations is $\Otil(\sqrt{r})$. In each PCG iteration, 
we solve the system in $\Lup_{1,\calT}$ via \nd.
We prove Lemma \ref{lem:solve_schur} in Section \ref{sect:schur_solve}.

Given the above lemmas, we prove Lemma \ref{lem:up_solver}.

\begin{proof}[Proof of Lemma \ref{lem:up_solver}]
The correctness of Algorithm \ref{alg:solver} is by Lemma \ref{lem:combine_solvers}.
By Lemma \ref{lem:solve_interior} and \ref{lem:solve_schur},
the total runtime of the algorithm is 
\begin{align*}
    \Otil \left(nr^{5/6} + n^{4/3}r^{5/18} +  nr + n^2 r^{-2/3} \right).
\end{align*}
\end{proof}

\subsection{Projection onto the Image of Up-Laplacian}

As the first Betti number of $\calK$ can be arbitrary, the approximate projection operators for the up $1$-Laplacian provided in \cite{CMFNPW14,BMNW22,BN22} are not applicable here. Our approximate projection operator follows a similar approach to our up-Laplacian solver, which is based on an incomplete \nd\ for \emph{triangles}, instead of edges. 

\begin{lemma}[Up-projection operator]
  Let $\calK$ be a pure $3$-complex embedded in $\R^3$ and composed of $n$ stable simplexes. Suppose we are given an $r$-hollowing for $\calK$. Then, for any $\eps > 0$, there exists an operator $\aPup_1$ such that for any $\bb$, $\aPup_1 \bb$ is in the image of $\Lup_1$, and
  \[
  \norm{\aPup_1 \bb - \Pup_1 \bb}_2 \le \eps \norm{\Pup_1 \bb}_2.
  \]
  In addition, $\aPup_1 \bb$ can be computed in time $O \left( nr + n^{4/3}r^{5/18} \log(n/\eps) + n^2 r^{-2/3} \right)$.
  \label{lem:up_projection_approx}
\end{lemma}

A detailed proof of Lemma \ref{lem:up_projection_approx} is deferred to Section \ref{sect:up_proj}.
The subsequent Lemma offers a helpful formula of $\Pup_1$, the orthogonal projection matrix onto the image of $\Lup_1$.

\begin{lemma}
Let $\calK$ be a simplicial complex with boundary operator $\partial_2$. For any partition $F \cup C$ of the $2$-simplexes of $\calK$, the orthogonal projection $\Pup_1$ for 
$\calK$ can be decomposed as
\begin{align*}
    \Pup_1 = \PPi_{\im(\partial_2[:,F])} 
+ \PkerF \partial_2[:,C] (\schur[\Ldown_2]_C)^{\dagger} \partial_2^\top [C,:] \PkerF,
\label{eqn:up_proj_up_Lap_formula}
\end{align*}
where $\Ldown_2$ is the down $2$-Laplacian.
\label{lem:many_balls_exact_up_proj}
\end{lemma}

Once more, an $r$-hollowing offers a natural partition of the triangles within $\calK$. We assign all the interior triangles to $F$ and all the boundary triangles to $C$. As such, \nd\ can be utilized to compute $\PPi_{\im(\partial_2[:,F])}$ and $\PkerF$, and PCG to solve the system in the Schur complement. 
The primary technical challenge arises 
when bounding the relative condition number of the Schur complement and the preconditioner, which requires a different approach.

\subsection{Finding \texorpdfstring{$r$}{Lg}-Hollowings}

\begin{lemma}[Finding an $r$-hollowing]
Let $\calK$ be a pure $3$-complex embedded in $\R^3$ and composed of $n$ stable simplexes each of volume $\Theta(1)$. 
If $\calK$ has a well-shaped boundary structure with parameter $r=o(1)$ defined as in Definiton \ref{def:hole_requirement},
then we can find an $r$-hollowing of $\calK$ in time $O(n\log n)$.
\label{lem:rhollowing}
\end{lemma}

Algorithm \ref{alg:hollow} describes how to find an $r$-hollowing in nearly-linear time. The algorithm first finds a \emph{nice bounding box} for the convex hull of $\calK$, then ``cuts'' the box into $O(n/r)$ smaller boxes of equal volume, and finally turns each smaller box into a region of an $r$-hollowing.

To clarify the importance of the well-shaped boundary structure, as defined in Definition \ref{def:hole_requirement}, we highlight its specific purposes as follows:
(\ref{asp:1}) The requirement for an $O(1)$ aspect ratio of the convex hull of $\calK$ is essential for obtaining a nice bounding box of $\calK$ with a comparable volume (Corollary \ref{lem:nice_bounding_box});
(\ref{asp:2}) The constraint on the $O(r^{1/3})$ diameter of boundary components ensures that each region's boundary forms a spherical shell;
(\ref{asp:3}) The bounding of the number of exterior simplexes guarantees the required bound on the number of boundary/exterior simplexes of each region of the resulting $r$-hollowing;
(\ref{asp:4}) The specification for the distance between boundary components guarantees a sufficiently wide shell width for the resulting $r$-hollowing. 

We prove Lemma \ref{lem:rhollowing} in Section \ref{sect:hollowing}.

%% file: solveF.tex
%!TEX root = main.tex

% \def\LupF{\widehat{\Lup_1}}
\def\LupF{\MM}
\def\LupFlarger{\widehat{\Lup_1}}

\section{Solver for Down-Laplacian}
\label{sect:down_solver}

This section shows how to solve $\Ldown_1 \xdown = \bb$ for any $\bb \in \im(\Ldown_1)$ in linear time and proves Lemma \ref{lem:down_lap_solver}. Recall that $\Ldown_1 = \partial_1^\top \WW_0 \partial_1$.
Our down-Laplacian solver works for \emph{any} simplicial complexes and returns a solution \emph{without error}.
Our approach is a slight modification of the down-Laplacian solver in \cite{CMFNPW14} to incorporate the vertex weights $\WW_0$.
Specifically, we compute $\xdown$ by three steps: 
solve $\partial_1^\top \WW_0^{1/2} \yy = \bb$, project $\yy$ onto $\im(\WW_0^{1/2} \partial_1)$ and get a new vector $\yy'$, then solve $\WW_0^{1/2} \partial_1 \xdown = \yy'$.
The first and the last steps can be solved by the approach in Lemma 4.2 of \cite{CMFNPW14} (stated below), and the second step can be explicitly solved by utilizing that $\partial_1$ is a vertex-edge incidence matrix of an oriented graph.

\begin{lemma}[A restatement of Lemma 4.2 of \cite{CMFNPW14}]
Given any $\bb_1 \in \im(\partial_1^{\top})$ (respectively, $\bb_0 \in \im(\partial_1)$), there is a linear operator $\partial_1^{+ \top}$ such that $\partial_1^\top \partial_1^{+ \top} \bb_1 = \bb_1$ (respectively, $\partial_1^+ = (\partial_1^{+ \top})^\top$
such that $\partial_1 \partial_1^{+ } \bb_0 = \bb_0$). In addition, we can compute $\partial_1^{+ \top} \bb_1 $ (respectively, $\partial_1^{+} \bb_0 $) in linear time.
\label{lem:cohen_down_solve}
\end{lemma}

We remark that the operators $\partial_1^+$ and $\partial_1^{+ \top}$ in Lemma \ref{lem:cohen_down_solve}
are not necessary to be the pseudo-inverses of $\partial_1$ and $\partial_1^\top$.

Claim \ref{clm:down_ker} explicitly characterizes the image of $\WW_0^{1/2} \partial_1$.

\begin{claim}
Let $\calK$ be a simplicial complex whose $1$-skeleton is connected, and let $v$ be the number of vertices in $\calK$.
Let $\WW_0 = \diag\left(w_1, \ldots, w_{v} \right)$ where $w_1,\ldots, w_{v} > 0$. Then, 
    \[
    \ker(\partial_1^\top \WW_0^{1/2}) = \Span \left\{ \uu \right\} 
    \text{ where } \uu = \left(\frac{1}{\sqrt{w_1}}, \ldots, \frac{1}{\sqrt{w_{v}}} \right)^\top.
    \]
\label{clm:down_ker}
\end{claim}

\begin{proof}
Let $\dims(\mathcal{V})$ be the dimension of a space $\mathcal{V}$.
Since all the diagonals of $\WW_0$ are positive, 
\begin{align*}
\dims(\ker(\partial_1^\top \WW_0^{1/2})) = 
\dims(\ker(\partial_1^\top ))
= v - \dims(\im(\partial_1))
= 1,
\end{align*}
where the last equality holds since the $1$-skeleton of $\calK$ is connected.
We can check that 
\begin{align*}
\partial_1^\top \WW_0^{1/2} \uu 
& = \partial_1^\top {\bf 1} = {\bf 0}.
\end{align*}
Here, ${\bf 1}$ is the all-one vector, and the second equality holds since $\partial_1$ is the vertex-edge incidence matrix of an oriented (weakly) connected graph.
Thus, the statement holds.
\end{proof}

We construct our down-Laplacian solver by combining Lemma \ref{lem:cohen_down_solve} and Claim \ref{clm:down_ker}. We remark this down-Laplacian solver applies to any simplicial complex.

\begin{proof}[Proof of Lemma \ref{lem:down_lap_solver}]
Without loss of generality, we assume the $1$-skeleton of $\calK$ is connected. Otherwise, we can write $\Ldown_1$ as a block diagonal matrix, each corresponding to a connected component of the $1$-skeleton, and we reduce solving a system in $\Ldown_1$ into solving several smaller down-Laplacian systems.

Let $\partial_1^+, \partial_1^{+ \top}$ be the linear operators in Lemma \ref{lem:cohen_down_solve} for $\calK$, and let $\uu$ be the vector in Claim \ref{clm:down_ker}.
We define 
\begin{align*}
% \label{eq:down_solver}
    \Zdown_1 \defeq \partial_1^+ \WW_0^{-1/2} \left(\II - \frac{\uu \uu^\top}{\norm{\uu}_2^2} \right) \WW_0^{-1/2} \partial_1^{+ \top}.
\end{align*}
We can compute $\Zdown_1 \bb$ in linear time.
In addition, we define
\begin{align*}
\yy = \partial_1^{+ \top} \bb, 
\zz = \WW_0^{-1/2} \yy, \zz_1 =  \left(\II - \frac{\uu \uu^\top}{\norm{\uu}_2^2} \right) \zz, 
\text{ and } \xx = \partial_1^+ \WW_0^{-1/2} \zz_1.
\end{align*}
Then, $\zz_1 \in \im(\WW_0^{1/2} \partial_1)$ and $\zz - \zz_1 \in \ker(\partial_1^\top \WW_0^{1/2})$.
\begin{align*}
\Ldown_1 \Zdown_1 \bb
& = \partial_1^\top \WW_0 \partial_1 \xx
= \partial_1^\top \WW_0 \partial_1 \partial_1^+ \WW_0^{-1/2} \zz_1 \\
    & = \partial_1^\top \WW_0 \WW_0^{-1/2} \zz_1
    \tag{since $\WW_0^{-1/2} \zz_1 \in \im(\partial_1)$}
    \\
    & = \partial_1^\top \WW_0^{1/2} \zz
    \tag{since $\zz - \zz_1 \in \ker(\partial_1^\top \WW_0^{1/2}) $} \\
    & 
    % = \partial_1^\top \WW_0^{1/2} \zz 
    = \partial_1^\top \yy = \bb. 
    % \qedhere
\end{align*}
\end{proof}

\section{Solver for Up-Laplacian}
\label{sec:up_solver_merge}

In this section, we will prove the two key lemmas, Lemma \ref{lem:solve_interior} and Lemma \ref{lem:solve_schur}, for building the first up-Laplacian solver.

\subsection{Solver for \texorpdfstring{$\mathbf{L}^{\text{up}}_1[F,F]$}{Lg}}
\label{sect:solver_int}

In this section, we construct an efficient solver $\upFsolve(\bb)$ that returns an 
$\xx$ such that  $\Lup_1[F,F] \xx = \bb$ for any $\bb \in \im(\Lup_1[F,F])$ and prove Lemma \ref{lem:solve_interior}.

We first show the interiors of different regions are ``disjoint'' in the sense that 
$\Lup_1[F,F]$ can be written as a block diagonal matrix where each diagonal block corresponds to the interior of a region after proper row and column permutation.
By definition, we can write $\Lup_1$ as the sum of rank-$1$ matrices that each corresponds to a triangle in $\calK$:
\begin{equation*}
\label{eq:L_as_a_sum}
    \Lup_1 = \partial_2 \WW_2 \partial_2^\top
= \sum_{\sigma: \text{triangle in } \calK} 
\WW_2[\sigma, \sigma] \cdot \partial_2[:, \sigma]  \partial_2[:, \sigma]^\top.
\end{equation*}
For any two edges $e_1, e_2$, $\Lup_1[e_1, e_2] = 0$ if and only if no triangle in $\calK$ contains both $e_1, e_2$.
We say such $e_1$ and $e_2$ are \emph{$\triangle$-disjoint}.
The following claim shows that interior edges from different regions are $\triangle$-disjoint.

\begin{claim}
Let $\calK$ be a stable $3$-complex and $\calT$ be an $r$-hollowing of $\calK$.
Let $R_1, R_2$ be two different regions of $\calK$ w.r.t. $\calT$, and let $e_1$ be an interior edge in $R_1$ and $e_2$ an interior edge in $R_2$.
Then, $e_1$ and $e_2$ are $\triangle$-disjoint.
\label{clm:hollow_interior_disjoint}
\end{claim}

\begin{proof}
Assume by contradiction there exists a triangle $\sigma \in \calK$ contains both $e_1, e_2$. Let $e_3$ be the third edge in $\sigma$.
Since $e_1, e_2$ are interior edges of different regions $R_1, R_2$, $e_3$ must cross the boundary of $R_1, R_2$. 
This contradicts the fact that regions can only intersect on their boundaries.
Thus, $e_1, e_2$ must be $\triangle$-disjoint.
\end{proof}

By the above Claim, 
computing an $\xx$ such that $\Lup_1[F,F] \xx = \bb$ reduces to computing a solution for each diagonal block submatrix per region.
For this reason, Lemma \ref{lem:solve_interior} is a corollary of the following Lemma \ref{lem:solve_interior_one_piece}.

\begin{lemma}
  Let $\calX$ be a $3$-complex with $O(r)$ simplexes embedded in $\R^3$ such that (1) each tetrahedron of $\calX$ has $O(1)$ aspect ratio and 
  (2) $\calX$ has $O(r^{2/3})$ exterior simplexes.
  Let $F$ be the set of interior edges of $\calX$, and let $\Lup_{1,\calX}$ be the up-Laplacian of $\calX$ and $\MM \defeq \Lup_{1,\calX}[F, F]$.
  Then, there is a permutation matrix $\PP$ and a lower triangular matrix $\LL$ with $O(r^{4/3})$ nonzeros such that 
  \begin{align}
      \MM = \PP \LL \LL^\top \PP^\top.
      \label{eqn:cholesky}
  \end{align}
  In addition, we can find such $\PP$ and $\LL$ in time $O(r^2)$. Given the above factorization, for any $\bb \in \im(\MM)$, we can compute an $\xx$ such that $\MM \xx = \bb$ in $O(r^{4/3})$ time.
  \label{lem:solve_interior_one_piece}
\end{lemma}

The factorization in Equation \eqref{eqn:cholesky} is called \emph{Cholesky factorization}. We will utilize the geometric structures of $\calX$ to find a sparse Cholesky factorization efficiently.
Miller and Thurston \cite{MT90} studied vertex separators of the $1$-skeleton of a $3$-complex $\calX$ in which each tetrahedron has $O(1)$ aspect ratio.
A subset of vertices $S$ of a graph over $v$ vertices \emph{$\delta$-separates} if the remaining vertices can be partitioned into two sets $A, B$ such that there are no edges between $A$ and $B$, and $\abs{A}, \abs{B} \le \delta v$. 
The set $S$ is an \emph{$f(v)$-separator} if there exists a constant $\delta < 1$ such that $S$ $\delta$-separates and $\abs{S} \le f(v)$.
These separators can be incorporated with \nd \ to efficiently compute a sparse Cholesky factorization of a matrix in which the nonzero structure encodes the $1$-skeleton of $\calX$.

\begin{theorem}[Vertex separator for a $3$-complex, Theorem 1.5 of \cite{MT90} and \cite{MTTV98}]
Let $\calX$ be a 3-complex in $\mathbb{R}^3$ in which each tetrahedron has $O(1)$ aspect ratio, and suppose $\calX$ has $t$ tetrahedrons and $\bar{v}$ exterior vertices.
Then, the $1$-skeleton of $\calX$ has a $O(t^{2/3} + \bar{v})$-separator that $4/5$-separates $\calX$.
In addition, such a separator can be found in linear time. 
\label{thm:mt90}
\end{theorem}

We need a slightly modified version of Theorem \ref{thm:mt90} to apply to our matrix $\LupF$.

\begin{corollary}[Edge separator for a $3$-complex]
Let $\calX$ be a $3$-complex that satisfies the requirements in Theorem \ref{thm:mt90}. In addition, the $1$-skeleton of $\calX$ is connected.
Let $m$ be the number of edges in $\calX$.
Then, there exists an algorithm that 
removes $O(t^{2/3} + \bar{v})$ edges in linear time so that the remaining edges can be partitioned into two $\triangle$-disjoint sets each of size at most $c m$ for some constant $c < 1$.
\label{cor:mt90}
\end{corollary}

\begin{proof}
Let $v = O(r)$ be the number of vertices in $\calX$.
By Lemma \ref{lem:max_deg}, $v = \Theta(m)$.
Let $S$ be the set of $O(t^{2/3} + \bar{v})$ vertices that $4/5$-separates $\calX$ (by Theorem \ref{thm:mt90}), and let $A,B$ be two disjoint sets of the remaining vertices after removing $S$ such that $\abs{A}, \abs{B} \le 4v/5$.
Let $E$ be the set of edges in $\calX$ incident to some vertex in $S$.
Since each vertex has $O(1)$ degree by Lemma \ref{lem:max_deg}, $\abs{E} = O(\abs{S}) =  O(t^{2/3} + \bar{v})$.
Now, we remove edges in $E$ from the edges of $\calX$.
Let $E_A$ be the set of the remaining edges incident to a vertex in $A$, and let $E_B$ be the set of the remaining edges incident to a vertex in $B$.
Since $A,B$ are disjoint, $E_A$ and $E_B$ are $\triangle$-disjoint.
We then show that $\abs{E_A}, \abs{E_B} \le cm$ for some constant $c < 1$.
Let $E'_A$ be the set of edges in $\calX$ that are incident to some vertex in $A$.
Since the $1$-skeleton of $\calX$ is a connected graph, $\abs{E'_A} \ge \abs{A} \ge (1-o(1)) \frac{v}{5}$.
Besides, $\abs{E'_A} = \abs{E_A} + \abs{E''_A}$, where $E''_A$ contains edges with one endpoint in $A$ and the other in $S$.
Since $\abs{E''_A} = O(t^{2/3} + \bar{v})$, we know $\abs{E_A} \ge (1-o(1)) \frac{v}{5} > c' m$ for some constant $c'$. 
Since $E_A$ and $E_B$ are disjoint,
we know $\abs{E_B} \le m - \abs{E_A} \le (1-c') m$. 
By symmetry, we have $\abs{E_A} \le (1-c') m$.
\end{proof}

We need the following theorem about \nd \ from \cite{LRT79}.

\begin{theorem}[\nd, Theorem 6 of \cite{LRT79}]
Let $\calG$ be any class of graphs closed under subgraph on which an $v^{\alpha}$-separator exists for $\alpha > 1/2$. 
Let $\AA \in \R^{v \times v}$ be symmetric and positive definite (that is, all eigenvalues of $\AA$ are positive). Let $G_A$ be a graph over vertices $\{1,\ldots,v\}$ where vertices $i,j$ are connected if and only if $\AA[i,j] \neq 0$.
If $G_A \in \calG$, then we can find a permutation matrix $\PP$ and a lower triangular matrix $\LL$ with $O(v^{2\alpha})$ nonzeros in $O(v^{3\alpha})$ time such that $\AA = \PP \LL \LL^\top \PP^\top$.
\label{thm:LRT79}
\end{theorem}

Combining Corollary \ref{cor:mt90} and Theorem \ref{thm:LRT79}, we prove Lemma \ref{lem:solve_interior_one_piece}.

\begin{proof}[Proof of Lemma \ref{lem:solve_interior_one_piece}]
Our approach for finding a sparse Cholesky factorization is the same as Section $4$ of \cite{MT90}. Specifically,
we find a family of separators for $\calX$ by recursively applying Corollary \ref{cor:mt90} to the remaining sets of edges $A, B$.
Since the separator size grows as a function of the exterior vertices, 
at the top level of the recursion, we include all the exterior edges $O(r^{2/3})$ in the root separator, which only increases the root separator by a constant factor.
In each remaining recursion step, suppose we want to separate $\calY$, a sub-complex of $\calX$; we let $\bar{\calY}$ be the complex consisting $\calY$ and all the simplexes in $\calX$ that contain a vertex in $\calY$.
Then we apply a slightly modified version of Corollary \ref{cor:mt90} (obtained by a slightly modified Theorem \ref{thm:mt90}) to $\bar{\calY}$ in which we only separate $\calY$. Here, the exterior-vertex term $O(\bar{v})$ in the size of the separator can be ignored since all the exterior edges of $\bar{\calY}$ have already been included in upper-level separators; we also use the fact that the number of exterior vertices and the number of exterior edges are equal up to a constant factor.
This separator family provides an elimination ordering for the edges of $\calK$, which is the permutation matrix $\PP$ in Equation \eqref{eqn:cholesky}, and the ordering uniquely determines the matrix $\LL$.
By Theorem \ref{thm:LRT79}, $\LL$ has $O(r^{4/3})$ nonzeros, and $\PP, \LL$ can be found in time $O(r^2)$.

One issue left is that $\MM$ in Lemma \ref{lem:solve_interior_one_piece} is positive semidefinite but \emph{not} positive definite. 
During the process of numeric factorization, the first row and column of some Schur complements are all-zero. 
We simply ignore these zeros and proceed (Ref: Chapter 4.2.8 of \cite{VG96}).
This produces a Cholesky factorization of $\PP^\top \LupF \PP = \LL \LL^\top$ such that only $k$ columns of $\LL$ are nonzero, where $k$ is the rank of $\LupF$.
We permute the rows and the columns of $\LL$ by multiplying permutation matrices $\PP_1, \PP_2$ so that 
\[
\PP_1 \LL \PP_2 = \begin{pmatrix}
\TT & {\bf 0}
\end{pmatrix} \defeq
\begin{pmatrix}
\HH_1 & {\bf 0} \\
\HH_2 & {\bf 0}
\end{pmatrix},
\]
where $\HH_1$ is a $k \times k$ non-singular lower-triangular matrix.
Then, solving $\MM \xx = \bb$ is equivalent to solving 
\[
\TT \TT^\top \zz = \PP_1 \PP^\top \bb \defeq \ff, 
~ \PP_1 \PP^\top \xx = \zz.
\]
We solve the first system by solving $\TT \yy = \ff$ and $\TT^\top \zz = \yy$.
We let $\yy$ satisfy $\HH_1 \yy = \ff[1:k]$. Since $\HH_1$ has full rank and $O(r^{4/3})$ nonzeros, such a $\yy$ exists and can be found in $O(r^{4/3})$ time.
Since $\yy \in \im(\TT)$, we know $\TT \yy = \ff$. 
Then we let $\zz[k+1:v] = 0$, where $v$ is the number of vertices in  $\calX$, and we let $\HH_1^\top \zz[1:k] = \yy$. Again, such $\zz$ exists and can be found in $O(r^{4/3})$ time. Finally, we compute $\xx =  \PP \PP_1^\top \zz$ in linear time.
\end{proof}

\begin{proof}[Proof of Lemma \ref{lem:solve_interior}]
We apply Lemma \ref{lem:solve_interior_one_piece} to each diagonal block matrix of $\Lup_1[F,F]$ to compute a sparse Cholesky factorization as in Equation \eqref{eqn:cholesky} in time 
\[
O \left( \frac{n}{r} \cdot r^2 \right) = O(n r).
\]
Given this Cholesky factorization, we can solve a system in $\Lup_1[F,F]$ in time 
\[
O \left( \frac{n}{r} \cdot r^{4/3} \right) = O\left(n r^{1/3} \right).
\]
\end{proof}

%% file: solve_schur.tex
%!TEX root = main.tex

% \def\LupOneT{\widehat{\Lup_{1,\calT}}}
\def\LupOneT{\MM}

\subsection{Solver for the Schur Complement}
\label{sect:schur_solve}

In this section, we establish a fast approximate solver for $\schur[\Lup_1]_C$ and prove Lemma \ref{lem:solve_schur}.
Recall that $C$ contains all the edges in $\calT$ (an $r$-hollowing of $\calK$), and $\schur[\Lup_1]_C$ is the Schur complement of the up-Laplacian operator of $\calK$ onto $C$. 
The idea is to run the \pcg \ (PCG) for systems in $\schur[\Lup_1]_C$ with preconditioner $\Lup_{1,\calT} = \partial_{2,\calT}\partial_{2,\calT}^\top$, which is the first up-Laplacian operator of $\calT$.
By Theorem \ref{thm:PCG}, the number of PCG iterations is $O(\sqrt{\kappa} \log(\kappa / \eps))$ where $\kappa = \kappa(\schur[\Lup_1]_C, \Lup_{1,\calT})$ is the relative condition number and $\eps$ is the error parameter; 
in each PCG iteration, we need to solve a system in $\Lup_{1,\calT}$, multiply $\schur[\Lup_1]_C$ with $O(1)$ vectors, and implement $O(1)$ vector operations.
In Section \ref{sect:condNum}, we upper bound the relative condition number.
In Section \ref{sec:proof_pcg_lemma}, we prove Lemma \ref{lem:solve_schur} about the solver for Schur complement.

We will need the following observation.
\begin{claim}
Let $\calX$ be a simplicial complex. 
Changing the orientations of the triangles in $\calX$ does not change its first up-Laplacian operator.
\label{clm:orient_triangle}
\end{claim}

\begin{proof}
Let $\partial_{2,\calX}$ be the second boundary operator of $\calX$, $\WW_{2,\calX}$ the diagonal matrix for the triangle weights, and $\Lup_{1,\calX}$ the first up-Laplacian.
Changing the orientations of the triangles in $\calX$  corresponds to multiplying a $\pm 1$ diagonal matrix $\XX$ to the right of $\partial_{2,\calX}$.
Observe
\[
\partial_{2,\calX} (\XX \WW_{2,\calX} \XX) \partial_{2,\calX}^\top = \partial_{2,\calX} \WW_{2,\calX} \partial_{2,\calX}^\top = \Lup_{1,\calX}.
\]
Thus the statement holds.
\end{proof}

\subsubsection{Preconditioning the Schur Complement}
\label{sect:condNum}

We will upper bound the relative condition number of $\schur[\Lup_1]_C$ and $\Lup_{1,\calT}$ and prove will Lemma \ref{lem:eigS}.

We decompose the $r$-hollowing $\calT$ into two parts. Recall that the boundary of each region triangulates a spherical shell in $\R^3$. Let $\calB_1$ be $2$-complex of the union of all the inner spheres of the boundaries, and let $\calB_2$ be $2$-complex of the rest of the boundaries. Let $C_1$ be the set of edges of $\calB_1$.
Let $\widehat{\Lup_1}$ be the first up-Laplacian of the union of $\calB_1$ and all the interior simplexes in all the regions. Then, 
\begin{align*}
\Lup_{1,\calT} = \Lup_{1,\calB_1} + \Lup_{1,\calB_2} 
\text{ and }
\schur[\Lup_1]_C = \schur[\widehat{\Lup_1}]_{C_1} + \Lup_{1,\calB_2}.
\end{align*}

\begin{claim}
If $\Lup_{1,\calB_1} \pleq \schur[\widehat{\Lup_1}]_{C_1} \pleq \alpha \Lup_{1,\calB_1}$ where $\alpha > 1$, then 
$\Lup_{1,\calT} \pleq \schur[\Lup_1]_{C} \pleq \alpha \Lup_{1,\calT}$.
\label{clm:schur_decompose_approx}
\end{claim}
\begin{proof}
\begin{align*}
& \schur[\Lup_1]_{C} 
\pgeq \Lup_{1,\calB_1} + \Lup_{1,\calB_2} = \Lup_{1,\calT}, \\
& \schur[\Lup_1]_{C} \pleq \alpha \Lup_{1,\calB_1} + \Lup_{1,\calB_2} \pleq \alpha \Lup_{1,\calT}.
\end{align*}
\end{proof}

To prove Lemma \ref{lem:eigS}, it suffices to show the following lemma.

\begin{lemma}
$\Lup_{1,\calB_1} \pleq \schur[\widehat{\Lup_1}]_{C_1} \pleq O(r) \Lup_{1,\calB_1}$.
\label{lem:eigS_inner_spheres}
\end{lemma}

The first inequality of Lemma \ref{lem:eigS_inner_spheres} follows immediately from the following well-known fact about Schur complements. One can find its proof from Fact 4.6 in \cite{KZ20}.

\begin{fact}
Let $\AA$ be a symmetric and PSD matrix:
\[
\AA = \begin{pmatrix}
\AA[F,F] & \AA[F,C] \\
\AA[C,F] & \AA[C,C]
\end{pmatrix}.
\]
Let $\schur[\AA]_C = \AA[C,C] - \AA[C,F] \AA[F,F]^{\dagger} \AA[F,C]$ be the Schur complement of $\AA$ onto $C$.
Then, for any $\xx \in \mathbb{R}^{\abs{C}}$, 
    \[
    \min_{\yy \in \mathbb{R}^{\abs{F}}}  \begin{pmatrix}
        \yy^\top & \xx^\top 
    \end{pmatrix} \AA \begin{pmatrix}
        \yy \\
        \xx
    \end{pmatrix} = \xx^\top \schur[\AA]_C \xx.
    \]
    \label{fact:schur}
As a corollary, $\schur[\AA]_C$ is symmetric and PSD.
\end{fact}

In the rest of the section, we prove the second inequality in Lemma \ref{lem:eigS_inner_spheres}.
Since both $\widehat{\Lup_1}$ and $\Lup_{1,\calB_1}$ can be written as a block diagonal matrix where each block corresponds to a region w.r.t. $\calT$ (after proper row and column permutation), it suffices to show the inequality in Lemma \ref{lem:eigS_inner_spheres} holds for each region, restated in the following lemma.

\begin{lemma}
Consider an $r$-hollowing region. Let $\calB$ be the $2$-complex of the inner sphere of the boundary, and let $B$ be the set of edges in $\calB$.
Let $\calX$ be $2$-complex of the union of $\calB$ and the interior simplexes.
Then, $\schur[\Lup_{1,\calX}]_{B} \pleq O\left(r \right) \Lup_{1,\calB}$.
\label{lem:eigS_single}
\end{lemma}

We first show that the images of $\schur[\Lup_{1,\calX}]_C$ and $\Lup_{1,\calB}$ are equal.

\begin{claim}
\label{clm:image}
  $\im(\schur[\Lup_{1,\calX}]_B) = \im(\Lup_{1,\calB})$.
\end{claim}

\begin{proof}
In the proof, we drop the subscript $\calX$ to simplify our notations by letting $\Lup_1 = \Lup_{1,\calX}$ and $\partial_2 = \partial_{2,\calX}$. 
We let $A$ be the set of edges in $\calX$ but not in $B$, and let $V_A$ be the set of vertices in $\calX$ but not in $\calB$ and $V_B$ the set of vertices in $\calB$.

Since both the two matrices $\schur[\Lup_{1}]_B, \Lup_{1,\calB}$ are symmetric and PSD,  the statement in the claim is equivalent to $\ker(\schur[\Lup_{1}]_B) = \ker(\Lup_{1,\calB})$.
By Fact \ref{fact:schur}, 
\[
\Lup_{1,\calB} \pleq \schur[\Lup_1]_B
\implies 
\ker(\schur[\Lup_{1}]_B) \subseteq \ker(\Lup_{1,\calB}).
\]
It remains to show $\ker(\schur[\Lup_{1}]_B) \supseteq \ker(\Lup_{1,\calB})$.

Let $\xx_B$ be an arbitrary vector in $\ker(\Lup_{1,\calB})$. We want to show $\xx_B \in \ker(\schur[\Lup_{1}]_B)$, that is,
$\xx_B^\top \schur[\Lup_{1}]_B \xx_B = 0$.
  By Fact \ref{fact:schur}, it suffices to show that there exists an $\xx = \begin{pmatrix}
      \xx_A \\
      \xx_B
\end{pmatrix}$ such that $\xx^\top \Lup_{1} \xx = 0$.
 This is equivalent to $\xx \in \ker(\partial_{2}^\top)$.
Suppose we add tetrahedrons, triangles, and necessary edges to $\calX$ and get a new simplicial complex $\calX'$ so that $\calB$ is the boundary of $\calX'$ and $\calX'$ triangulates a $3$-ball. 
Since operator $\partial_2^\top$ maps the vector space of edges to the vector space of triangles and no new edges in $\calX'$ can appear in a triangle in $\calX$, it suffices to show there exists an $\xx' = \begin{pmatrix}
    \xx'_A \\
    \xx_B
\end{pmatrix}$ such that 
$\xx' \in \ker(\partial_{2,\calX'}^\top)$ (restricting $\xx'_A$ to $\xx_A$ gives $\xx \in \ker(\partial_{2,\calX}^\top)$).
Since the first Betti number of $\calX'$ is zero, 
 \begin{align}
\xx \in \ker(\partial_{2,\calX'}^\top) \iff \xx \perp \im(\partial_{2,\calX'}) 
\iff \xx \in \im(\partial_{1,\calX'}^\top),
\label{eqn:ker_im}
 \end{align}
Let $A'$ be the set of interior edges in $\calX'$, and $V_{A'}$ the set of interior vertices in $\calX'$.
  We can write 
  \[
    \partial_{1,\calX'}^\top = \begin{pmatrix}
      \partial_{1,\calX'}^\top[A',V_{A'}] & \partial_{1,\calX'}^\top[A',V_B] \\
      {\bf 0} & \partial_{1,\calX'}^\top[B,V_B] 
    \end{pmatrix},
  \]
  where $\partial_{1,\calX'}[V_B,B] = \partial_{1,\calB}$.
  Since $\xx_B \in \ker(\partial_{2,\calB}^\top)$ and the first Betti number of $\calB$ is $0$, by an argument similar to Equation \eqref{eqn:ker_im}, we have 
  $\xx_B \in \im(\partial_{1,\calB}^\top)$, that is, $\xx_B = \partial_{1,\calB}^\top \yy_B$ for some $\yy_B$.
  Setting $\xx_A' = \partial_{1}^\top[A',V_B] \yy_B$, we have $\xx \in \im(\partial_{1,\calX'}^\top)$.
\end{proof}

\begin{claim}
\label{clm:eigmax}
  $\lambda_{\max}(\schur[\Lup_{1,\calX}]_B) = O(w_{\max})$, where $w_{\max}$ is the maximum triangle weight in $\calX$.
\end{claim}

\begin{proof}
We drop the subscript $\calX$ in the proof to simplify our notations.
The Courant-Fischer Minimax Theorem (Ref: Theorem 8.1.2 of \cite{VG96}) states that for any symmetric matrix $\AA$, 
\[
\lambda_{\max}(\AA) = \max_{\xx: \norm{\xx}_2 = 1} \xx^\top \AA \xx.
\]
Apply this theorem and Fact \ref{fact:schur},
  \[
  \lambda_{\max} (\schur[\Lup_{1}]_B) = 
  \max_{\xx: \norm{\xx}_2 = 1} \xx^\top  \schur[\Lup_{1}]_B \xx 
  \le \max_{\xx: \norm{\xx}_2 = 1} \begin{pmatrix}
    {\bf 0}^\top &  \xx^\top
  \end{pmatrix} \Lup_{1} \begin{pmatrix}
    {\bf 0} \\
    \xx
  \end{pmatrix}
  \le \lambda_{\max}(\Lup_{1}).
  \]
Below, we bound $\lambda_{\max}(\Lup_1)$:
\begin{multline*}
\lambda_{\max}(\Lup_1) = \max_{\xx: \norm{\xx}_2 = 1} \xx^\top \Lup_1 \xx
= \max_{\xx: \norm{\xx}_2 = 1} \sum_{\sigma: \text{triangle in }\calX} 
\WW_2[\sigma, \sigma] (\partial_2[:,\sigma]^\top \xx)^2 \\
\le 3 w_{\max} \cdot \max_{\xx: \norm{\xx}_2 = 1} 
\sum_{\substack{\sigma = [i,j,k]: \\ \text{triangle in }\calX}} (\xx[i]^2 + \xx[j]^2 + \xx[k]^2) 
= O(w_{\max}).
\end{multline*}
The last inequality holds since each edge appears in at most $O(1)$ triangles by Lemma \ref{lem:max_deg}.
Thus, $\lambda_{\max}(\schur[\Lup_{1}]_B) = O(w_{\max})$.
\end{proof}

One more piece we need is a lower bound for $\lambda_{\min}(\Lup_{1,\calB})$, which will be established via eigenvalues of graph Laplacian matrices.

\begin{theorem}[Section 4.2 of \cite{M91}]
\label{thm:diameter}
  Let $G$ be an unweighted graph over $n$ vertices with diameter $\diam$, the length of the longest path in $G$. Let $\LL_G$ be the graph Laplacian matrix of $G$.
  Then, 
  \[
  \lambda_{\min} (\LL_G) \ge \frac{4}{n \diam}.  
  \]
\end{theorem}

\begin{claim}
  $\lambda_{\min}(\Lup_{1,\calB}) = \Omega( w_{\min}\cdot r^{-1} )$, where $w_{\min}$ is the minimum triangle weight in $\calB$.
  \label{clm:eigmin}
\end{claim}

\begin{proof}
We drop the subscript $\calB$ in the proof to simplify our notations.
Again we decompose the matrix as a sum of rank-1 matrices for each triangle in $\calB'$:
\begin{align*}
\Lup_{1} = \sum_{\sigma: \text{triangle in }\calB} \WW_2[\sigma, \sigma]
\partial_2[:,\sigma] \partial_2[:,\sigma]^\top
\pgeq w_{\min} \cdot  \partial_{2} \partial_{2}^\top.
\end{align*}
So,
\[
\lambda_{\min}(\Lup_1) 
\ge w_{\min} \cdot  \lambda_{\min} (\partial_2 \partial_2^\top) =
 w_{\min} \cdot  \lambda_{\min} (\partial_2^\top \partial_2).
\]
Since $\calB$ triangulates a two-sphere, each edge of $\calB$ appears in exactly two triangles.
Since changing the orientations of the triangles in $\calB$ does not change $\Lup_1$ (by Claim \ref{clm:orient_triangle}), we assume all the triangles in $\calB$ are oriented clockwise.
Then, each column of $\partial_2^\top$ has exactly one entry with value $1$ and one entry $-1$ and all others $0$. 
That is, $\partial_2^\top \partial_2$ is the Laplacian of the dual graph of $\calB$: the vertices are the triangles in $\calB$, and two vertices are adjacent if and only if the corresponding two triangles share a common edge.
The dual graph has $O(r^{2/3})$ vertices and diameter $O(r^{1/3})$. 
By Theorem \ref{thm:diameter}, $\lambda_{\min} (\partial_2^\top \partial_2)
= \Omega(w_{\min} \cdot r^{-1})$.
\end{proof}

Combining all the claims above, we prove Lemma \ref{lem:eigS_single}.

\begin{proof}[Proof of Lemma \ref{lem:eigS_single}]
Let $\PPi$ be the orthogonal projection matrix onto $\im(\Lup_{1,\calB}) = \im(\schur[\Lup_{1,\calX}]_{B})$ (by Claim \ref{clm:image}).
By Claim \ref{clm:eigmax}, $\schur[\Lup_{1,\calX}]_{B} \pleq   O(w_{\max}) \PPi $.
By Claim \ref{clm:eigmin}, $\PPi \pleq O(\frac{r}{w_{\min}} ) \Lup_{1,\calB}$. 
Combining all together, we have 
$\schur[\Lup_{1,\calX}]_{B} \pleq O(rU) \Lup_{1,\calB}$, where $U = \frac{w_{\max}}{w_{\min}} = O(1)$.
\end{proof}

\subsubsection{Proof of Lemma \ref{lem:solve_schur}}
\label{sec:proof_pcg_lemma}
By Theorem \ref{thm:PCG} and Lemma \ref{lem:eigS}, the number of \pcg \ (PCG) iterations is 
\[
\Otil\left( \kappa \left( \schur[\Lup_1]_C, \Lup_{1,\calT} \right)^{1/2}
\right) = \Otil(r^{1/2}).
\]
In each PCG iteration, we solve a system in $\Lup_{1,\calT}$ and multiply $\schur[\Lup_1]_C$ with $O(1)$ vectors.
Recall 
\[
\schur[\Lup_1]_C = \Lup_1[C, C] - \Lup_1[C, F] \Lup_1[F, F]^{\dagger} \Lup_1[F, C].
\]
In our preprocessing, we compute a Cholesky factorization of $\Lup_1[F, F]$ in time $O(r^2 \cdot \frac{n}{r}) = O(nr)$; then, we can multiply $\schur[\Lup_1]_C$ onto a vector in time $O(r^{4/3} \cdot \frac{n}{r}) = O( nr^{1/3})$.
Similarly, we solve a system in $\Lup_{1,\calT}$ by \nd. 
By our construction, $\calT$ has $O(r^{2/3} \cdot \frac{n}{r}) = O(nr^{-1/3})$ triangles.
The Cholesky factorization runs in time $O(n^2 r^{-2/3})$, and each system solve runs in time $O(n^{4/3} r^{-4/9})$.
So, the runtime per each PCG iteration is 
\[
O\left( nr^{1/3} + n^{4/3} r^{-4/9}
\right).
\]
Therefore, the total time is 
\begin{align*}
\Otil \left( r^{1/2} \left( nr^{1/3} + n^{4/3} r^{-4/9} \right)
+ nr + n^2 r^{-2/3}
\right)
= \Otil \left( nr + n^{4/3}r^{5/18} + n^2 r^{-2/3}
\right).
\end{align*}
This completes the proof of Lemma \ref{lem:solve_schur}.

%% file: up_project.tex
%!TEX root = main.tex

\def\SS{\boldsymbol{\mathit{S}}}
\def\upSolverUnweighted{\textsc{SolveUnweightedUpLap}}
\def\Pit{\widetilde{\PPi}}
\def\down2solve{\textsc{DownLap2Solver}}
\def\up1solve{\textsc{UpLap1Solver}}
\def\unionUpF{\textsc{UnionUpLapFSolver}}

\def\aPf{\widetilde{\PPi}_F}
\def\bbtil{\widetilde{\bb}}
\def\aPf{\widetilde{\PPi}_{\ker(\partial_2^\top[F,:])}}
\def\PimF{\PPi_{\im(\partial_2[:,F])}}

\def\one{{\bf 1}}

\section{Projection onto the Image of Up-Laplacian}
\label{sect:up_proj}

In this section, we will show how to approximately project a vector onto the image of the first up-Laplacian of a well-shaped $3$-complex with a given $r$-hollowing and prove Lemma \ref{lem:up_projection_approx}. 
The following lemma gives an explicit formula for the orthogonal projection onto the up-Laplacian.

\begin{lemma}
Let $\calK$ be a simplicial complex with boundary operator $\partial_2$. For any partition $F \cup C$ of the $2$-simplexes of $\calK$, the orthogonal projection $\Pup_1$ for 
$\calK$ can be decomposed as
\begin{align}
    \Pup_1 = \PPi_{\im(\partial_2[:,F])} 
+ \PkerF \partial_2[:,C] (\schur[\Ldown_2]_C)^{\dagger} \partial_2^\top [C,:] \PkerF.
\label{eqn:up_proj_up_Lap_formula}
\end{align}
In addition, the second matrix on the right-hand side of the above equation is the orthogonal projection matrix onto the image of 
$\PkerF \partial_2[:,C]$.
\label{lem:many_balls_exact_up_proj_restate}
\end{lemma}

\begin{proof}
For any $\bb \in \R^n$,
we define 
\begin{align}
\begin{split}
& \ff_F = \partial_2^\top[F,:] \bb, ~ \ff_C = \partial_2^\top[C,:] \bb \\
& \hh = \ff_C - \Ldown_2[C,F] (\Ldown_2[F,F])^{\dagger} \ff_F  \\
& \xx_C = (\schur[\Ldown_2]_C)^{\dagger} \hh  \\
& \xx_F = (\Ldown_2[F,F])^{\dagger} \left(
 \ff_F - \Ldown_2[F,C] \xx_C
\right).
\end{split}
\label{eqn:many_balls_x}
\end{align}
Let $\xx = \begin{pmatrix}
\xx_C \\
\xx_F
\end{pmatrix}$. Applying Lemma \ref{lem:combine_solvers} with $\delta = 0$, we have $\partial_2^\top \partial_2 \xx = \partial_2^\top \bb$. 
That is, $\xx$ can be written as $\xx = \xx_1 + \xx_2$ where $\xx_1 = (\partial_2^\top \partial_2)^{\dagger} \partial_2^\top \bb$ and $\xx_2$ is in $\ker(\partial_2)$. Then, $\partial_2 \xx = \partial_2 \xx_1 = \Pup_1 \bb$.

By Equation \eqref{eqn:many_balls_x},
\[
\Pup_1 \bb = 
\partial_2 \xx = \PPi_{\im(\partial_2[:,F])} \bb 
+ \PkerF \partial_2[:,C] (\schur[\Ldown_2]_C)^{\dagger}
\partial_2^\top [C,:] \PkerF \bb.
\]
Thus, the equation in the statement holds.
By Fact \ref{fact:schur_project_kernel},
\begin{align*}
\schur[\Ldown_2]_C = \partial_2^\top [C,:] \PkerF \partial_2 [:,C].
\label{eqn:many_balls_schur_proj_kernel}
\end{align*}
The second matrix on the right-hand side of the equation in the statement is the orthogonal projection onto the image of 
$\PkerF \partial_2[:,C]$.
\end{proof}

To apply $\Pup_1$ to a vector $\bb$, we will need to
(1) project $\bb$ onto $\im(\partial_2[:,F])$, (2) project $\bb$ onto $\ker(\partial_2^\top[F,:])$, 
and (3) solve a system in $\schur[\Ldown_2]_C$.
We let $F$ be the set of all the interior triangles of $\calK$ w.r.t. the given $r$-hollowing, and let $F$ be the set of all the boundary triangles.
Similar to our up-Laplacian solver, we will apply the \nd \ for the ``$F$" part and the \pcg \ for the Schur complement onto $C$.

With a slight modification of the edge separator for a $3$-complex (Corollary \ref{cor:mt90}), we obtain a triangle separator for a $3$-complex as the corollary below.

\begin{corollary}[Triangle separator for a $3$-complex]
Let $\calX$ be a $3$-complex that satisfies the requirements in Theorem \ref{thm:mt90} and has a connected $2$-skeleton.
Let $p$ be the number of triangles in $\calX$.
Then, there exists an algorithm that 
removes $O(t^{2/3} + \bar{v})$ triangle in linear time. The remaining triangles can be divided into two sets, each containing at most $cp$ triangles (where $c < 1$ is a constant), with no shared edges between the two sets.
\label{cor:triangle_separator}
\end{corollary}

\begin{proof}
The proof is very similar to that of Corollary \ref{cor:mt90}, by replacing edges with triangles.
\end{proof}

Given a triangle separator algorithm for a $3$-complex, we can obtain a solver for $\Ldown_2[F,F] = \partial_2^\top[F,:] \partial_2[:,F]$ by \nd. 
Note $\Ldown_2[\sigma_1, \sigma_2] \neq 0$ if and only if two triangles $\sigma_1, \sigma_2$ share a common edge.

\begin{lemma}
  Let $\calX$ be a $3$-complex with $O(r)$ simplexes embedded in $\R^3$ such that (1) each tetrahderon of $\calX$ has $O(1)$ aspect ratio and (2) $\calX$ has $O(r^{2/3})$ exterior simplexes.
  Let $F$ be the set of interior triangles of $\calX$, and let $\Ldown_{2,\calX}$ be the second down-Laplacian of $\calX$ and $\MM \defeq \Ldown_{2,\calX}[F,F]$.
  Then, there is a permutation matrix $\PP$ and a lower triangular matrix $\LL$ with $O(r^{4/3})$ nonzeros such that $\MM = \PP \LL \LL^\top \PP^\top$.
  In addition, we can find such $\PP$ and $\LL$ in time $O(r^2)$. Given the above factorization, for any $\bb \in \im(\MM)$, we can compute an $\xx$ such that $\MM \xx = \bb$ in $O(r^{4/3})$ time.
  \label{lem:solve_interior_one_piece_triangle}
\end{lemma}

\begin{claim}
Given any $1$-chain vector $\bb$, we can compute 
$\PPi_{\im(\partial_2[:,F])} \bb$ and $\PkerF \bb$ in time $O(nr)$.
\label{clm:up_proj_solve}
\end{claim}

\begin{proof}
We have 
\begin{align*}
\PPi_{\im(\partial_2[:,F])}\bb
= \partial_2[:,F] \left( \partial_2^\top[F,:] \partial_2[:,F]
\right)^{\dagger} \partial_2^\top[F,:] \bb.
\end{align*} 
We can compute $\partial_2^\top[F,:] \bb$ in time $O(n)$.
To apply $\left( \partial_2^\top[F,:] \partial_2[:,F]
\right)^{\dagger}$, we compute a Cholesky factorization of  $\partial_2^\top[F,:] \partial_2[:,F]$ in time $O(r^2 \cdot \frac{n}{r}) = O(nr)$ via \nd \ (Lemma \ref{lem:solve_interior_one_piece_triangle}), where the Cholesky factorization has $O(r^{4/3} \cdot \frac{n}{r}) = O(nr^{1/3})$ nonzeros;
given such a Cholesky factorization, we can apply $\left( \partial_2^\top[F,:] \partial_2[:,F]
\right)^{\dagger}$ to a vector in time $O(nr^{1/3})$.
Then, multiplying $\partial_2[:,F]$ with a vector runs in linear time. 
So, the total runtime of computing $\PPi_{\im(\partial_2[:,F])}\bb$ is $O(nr)$.
Since $\PkerF = \II - \PPi_{\im(\partial_2[:,F])}$, we can apply $\PkerF$ to a vector in the same time up to constant.
\end{proof}

It remains to bound the runtime of (approximately) solving a system in the Schur complement $\schur[\Ldown_2]_C$.
We run the \pcg \ 
and precondition the Schur complement by the boundary itself $\partial_2^\top[C,:] \partial_2[:,C]$.

\subsection{Preconditioning the Schur Complement}
\label{sec:up_proj_precondition}

In this section, we prove the following lemma for the relative condition number of the preconditioner for the Schur complement.

\begin{lemma}
$\kappa(\schur[\Ldown_2]_C, \partial_2^\top[C,:] \partial_2[:,C])= O(r)$. 
\label{lem:up_proj_precondition_whole}
\end{lemma}

We will bound the relative condition number for each region and then combine them to get Lemma \ref{lem:up_proj_precondition_whole}.
Recall that in each region of an $r$-hollowing, the boundary triangulates a spherical shell in $\R^3$.
We call the boundary triangles in a region containing an edge on the inner sphere of the boundary \emph{the boundary layer}.
Since each region boundary has its shell width of at least $5$, boundary layers from different regions are disjoint.
For the $i$th region, let $\SS_i$ be the Schur complement of the interior triangles and the triangles in the boundary layer of the region $i$ onto its boundary layer.
We can write the Schur complement 
\[
\schur[\Ldown_2]_C = \diag(\SS_1, \ldots, \SS_k)
+ \AA,
\]
where $\AA$ is a PSD matrix.
Similarly, let $\MM_i$ be the boundary layer of the region $i$. We write the boundary 
\[
\partial_2^\top[C,:] \partial_2[:,C]
= \diag(\MM_1, \ldots, \MM_k) + \AA.
\]
By an argument similar to Claim \ref{clm:schur_decompose_approx}, proving Lemma \ref{lem:up_proj_precondition_whole} reduces to proving the following lemma.

\begin{lemma}
For each $i$, $\Omega(r^{-1}) \PPi_i \MM_i \PPi_i \pleq \SS_i \pleq \MM_i.$
\label{lem:up_proj_schur_precond}
\end{lemma}

In the rest of this subsection, we prove Lemma \ref{lem:up_proj_schur_precond}. We locally use $\partial_2$ for the boundary operator \emph{of the region $i$} and drop the region index $i$ in the rest of this subsection. 
We write 
\[
\partial_2^\top = \begin{pmatrix}
  \BB_{int} \\
  \BB_{bd}
\end{pmatrix}  
\]
where the rows of $\BB_{int}$ correspond to the interior triangles 
and the rows of $\BB_{bd}$ the \emph{boundary layer} triangles.
The Schur complement onto the boundary layer triangles, by Fact \ref{fact:schur_project_kernel},
\[
\SS = \BB_{bd} \PPi_{\ker(\BB_{int})}  \BB_{bd}^\top.
\]
Let $\MM \defeq \BB_{bd} \BB_{bd}^\top$, and let $\PPi$ be the orthogonal projection onto the image of $\SS$. Then, $\SS \pleq \MM$, and 
\begin{align*}
\SS \pgeq \PPi \MM^{1/2} \MM^{\dagger /2} \SS 
\MM^{\dagger / 2} \MM^{1/2} \PPi
\pgeq \lambda_{\min}(\MM^{\dagger / 2} \SS \MM^{\dagger / 2}) \cdot \PPi \MM \PPi.
\end{align*}

\begin{claim}
\[
\lambda_{\min} \left( \MM^{\dagger /2} \SS \MM^{\dagger /2}
\right) = \lambda_{\min} \left( \PPi_{\im(\BB_{bd}^\top)} \PPi_{\ker(\BB_{int})} \PPi_{\im(\BB_{bd}^\top)} \right).
\]
\end{claim}

\begin{proof}
We take the singular value decomposition: $\BB_{bd} = \UU \DD \VV^\top$, where the columns of $\UU$ (resp., $\VV$) form an orthonormal basis of $\im(\BB_{bd})$ (resp., $\im(\BB^\top_{bd})$).
Then, 
\begin{align*}
  \MM^{\dagger /2} \SS \MM^{\dagger /2}
  & = \UU \DD^{-1} \UU^\top \UU \DD \VV^\top \PPi_{\ker(\BB_{int})}
  \VV \DD \UU^\top \UU \DD^{-1} \UU^\top 
   = \UU \VV^\top \PPi_{\ker(\BB_{int})} \VV \UU^\top.
\end{align*}
Since $\UU$'s columns are orthonormal, 
\[
\lambda_{\min} (\MM^{\dagger /2} \SS \MM^{\dagger /2}) = \lambda_{\min}(\VV^\top \PPi_{\ker(\BB_{int})} \VV).
\]
We claim 
\begin{align}
\lambda_{\min}(\VV^\top \PPi_{\ker(\BB_{int})} \VV)
= \lambda_{\min}(\VV \VV^\top \PPi_{\ker(\BB_{int})} \VV \VV^\top),
\label{eqn:orth_proj_eigs}
\end{align}
which implies the claim statement.
Let $\lambda$ be an eigenvalue of $\VV^\top \PPi_{\ker(\BB_{int})} \VV$ with eigenvector $\uu$. Then, 
\begin{align*}
\VV^\top \PPi_{\ker(\BB_{int})} \VV \uu = \lambda \uu
& \implies \VV \VV^\top \PPi_{\ker(\BB_{int})} \VV \VV^\top \VV \uu = \lambda \VV \uu \\
& \implies \PPi_{\im(\BB_{bd}^\top)} \PPi_{\ker(\BB_{int})} \PPi_{\im(\BB_{bd}^\top)} \VV \uu = \lambda \VV \uu.
\end{align*}
That is, $\lambda$ is an eigenvalue of $\PPi_{\im(\BB_{bd}^\top)} \PPi_{\ker(\BB_{int})} \PPi_{\im(\BB_{bd}^\top)} $ with eigenvector $\VV \uu$.
Let $\mu$ be an eigenvector of $\PPi_{\im(\BB_{bd}^\top)} \PPi_{\ker(\BB_{int})} \PPi_{\im(\BB_{bd}^\top)} $ with eigenvector $\ww$. Then, 
\begin{align*}
\PPi_{\im(\BB_{bd}^\top)} \PPi_{\ker(\BB_{int})} \PPi_{\im(\BB_{bd}^\top)} \ww = \mu \ww
& \implies \VV^\top \PPi_{\ker(\BB_{int})} \VV 
\VV^\top \ww = \mu \VV^\top \ww.
\end{align*}
That is, $\mu$ is an eigenvalue of $\VV^\top \PPi_{\ker(\BB_{int})} \VV $ with eigenvector $\VV^\top \ww$. So, Equation \eqref{eqn:orth_proj_eigs} holds.
\end{proof}

\begin{claim}
The image of $\PPi_{\im(\BB_{bd}^\top)} \PPi_{\ker(\BB_{int})} \PPi_{\im(\BB_{bd}^\top)} $ is 
$\calU = \{\uu \in \im(\BB_{bd}^\top): \uu \perp \im(\BB_{bd}^\top) \cap \im(\BB_{int}^\top)\}$.
\label{clm:proj_image}
\end{claim}
\begin{proof}
We first find an orthogonal basis of the kernel of $\PPi_{\im(\BB_{bd}^\top)} \PPi_{\ker(\BB_{int})} \PPi_{\im(\BB_{bd}^\top)}$, which is 
$\calU' \defeq \{\uu: \PPi_{\im(\BB_{bd}^\top)} \PPi_{\ker(\BB_{int})} \PPi_{\im(\BB_{bd}^\top)} \uu = {\bf 0}\} = \{\uu: \PPi_{\ker(\BB_{int})} \PPi_{\im(\BB_{bd}^\top)} \uu = {\bf 0}\}$.
Clearly, $\ker(\BB_{bd}) \subset \calU'$.
We consider $\uu \in \im(\BB_{bd}^\top)$. In this case, $\PPi_{\ker(\BB_{int})} \PPi_{\im(\BB_{bd}^\top)} \uu =  \PPi_{\ker(\BB_{int})} \uu = {\bf 0}$ if and only if $\uu \in \im(\BB_{int}^\top)$.
So, the claim statement holds.
\end{proof}

\begin{lemma}
$\lambda_{\min} \left( \PPi_{\im(\BB_{bd}^\top)} \PPi_{\ker(\BB_{int})} \PPi_{\im(\BB_{bd}^\top)} \right) = \Omega(r^{-1})$.
\end{lemma}

\begin{proof}
By the Courant-Fischer min-max theorem and Claim \ref{clm:proj_image},
\begin{align}
    \lambda_{\min} \left( \PPi_{\im(\BB_{bd}^\top)} \PPi_{\ker(\BB_{int})} \PPi_{\im(\BB_{bd}^\top)} \right)
= \min_{\uu \in \calU } \frac{\uu^\top \PPi_{\im(\BB_{bd}^\top)} \PPi_{\ker(\BB_{int})} \PPi_{\im(\BB_{bd}^\top)} \uu}{\uu^\top \uu},
\label{eqn:lambda_min_proj}
\end{align}
where $\calU = \{\uu \in \im(\BB_{bd}^\top): \uu \perp \im(\BB_{bd}^\top) \cap \im(\BB_{int}^\top) \} \setminus \{{\bf 0}\}$. 

Recall that in each region, the boundary triangulates a spherical shell in $\R^3$.
We further decompose $\partial_2$ according to the inner sphere of the boundary:
\begin{align}
  \partial_2 = \begin{pmatrix} 
    \BB_{int}^\top & | & \BB_{bd}^\top
  \end{pmatrix} = 
    \begin{pmatrix}
    \BB_{11} & \BB_{12} & |  & {\bf 0} \\
    {\bf 0} & \BB_2 & | & \BB_3 \\
    {\bf 0} & {\bf 0} & | & \BB_{4}
  \end{pmatrix}.
  \label{eqn:decompose_partial_2}
\end{align} 
Here, the blocks of the rows from top to bottom correspond to the interior edges, 
the boundary edges on the boundary inner sphere, 
and the other boundary edges, respectively;
the blocks of the columns from left to right correspond to the interior triangles with only 
interior edges, the interior triangles with both interior and boundary edges, and the boundary triangles, respectively.
The column in $\partial_2$ corresponding to a triangle whose edges are all on the boundary inner sphere can be written as a linear combination of the other boundary triangle columns. 
We remove these boundary inner sphere triangles so that every remaining triangle has at most one edge on the boundary inner sphere.
Without loss of generality, we can orient and reorder the edges and the triangles so that 
\begin{align}
  \BB_2 = \diag(\one^\top, \ldots, \one^\top),  
\BB_3 = \diag(\one^\top, \ldots, \one^\top).
\label{eqn:BB2}
\end{align}
In addition, in each triangle with one edges on the boundary inner sphere, we let the other two edges point away from the inner sphere so that 
every column of $\BB_{12}$ and of $\BB_{4}$ has exactly two nonzero entries with values $1$ and $-1$ each.

We want to characterize $\im(\BB_{int}^\top) \cap \im(\BB_{bd}^\top)$. Let 
\begin{align*}
   \calV \defeq \left\{ \begin{pmatrix}
  {\bf 0} \\
  \BB_3 \xx \\
  {\bf 0}
\end{pmatrix}: \xx \in \ker(\BB_{4})
\right\}.
\end{align*}
By Equation \eqref{eqn:decompose_partial_2}, $\im(\BB_{int}^\top) \cap \im(\BB_{bd}^\top) \subseteq \calV$. We will show $\calV \subseteq \im(\BB_{int}^\top) \cap \im(\BB_{bd}^\top)$, that is, $\calV \subseteq \im(\BB_{int}^\top)$.
Since the boundary layer triangulates a spherical shell in $\R^3$, whose first Betti number is zero, $\calV$ is orthogonal to the image of $\partial_1^\top$ of the boundary layer.
Without loss of generality, we can also assume the $2$-complex of the interior triangles touching the boundary inner sphere has the first Betti number being zero. Otherwise, we shift the boundary inner sphere towards the boundary outer sphere and include the boundary layer in the interior part. This can be done since, in the given $r$-hollowing, each region boundary triangulates a spherical shell with a ``hop'' shell width of at least $5$. Under this assumption, 
$\calV \subseteq \im(\left( \begin{array}{c}
     \BB_{12}  \\
     \BB_2 \\
     {\bf 0}
\end{array} \right)) \subseteq \im(\BB_{int}^\top)$.
Thus, $\calV = \im(\BB_{int}^\top) \cap \im(\BB_{bd}^\top)$.
Then, 
\begin{align*}
\calU = \left\{ \begin{pmatrix}
    {\bf 0} \\
    \BB_3 \yy \\
    \BB_{4} \yy
\end{pmatrix} : \BB_3 \yy \perp \BB_3 \xx, ~ \forall \xx \in \ker(\BB_{4})
\right\} \setminus \{{\bf 0}\}.
\end{align*}
By Equation \eqref{eqn:lambda_min_proj} and \eqref{eqn:decompose_partial_2},
\begin{align*}
\lambda_{\min} \left( \PPi_{\im(\BB_{bd}^\top)} \PPi_{\ker(\BB_{int})} \PPi_{\im(\BB_{bd}^\top)} \right)
& \ge \min_{\yy \neq {\bf 0}: \BB_3^\top \BB_3 \yy \perp  \ker(\BB_{4})} \frac{\norm{\BB_{4} \yy}^2}{\norm{\BB_3\yy}^2 + \norm{\BB_{4} \yy}^2} \\
& = \min_{\yy \neq {\bf 0}: \BB_3^\top \BB_3 \yy \perp  \ker(\BB_{4)}} \frac{1}{\norm{\BB_3\yy}^2 / \norm{\BB_{4} \yy}^2 + 1}.
\end{align*}
It suffices to show $\frac{\norm{\BB_3\yy}^2}{\norm{\BB_{4} \yy}^2} = O(r)$.

Without loss of generality, we can assume each diagonal block of $\BB_3$ in Equation \eqref{eqn:BB2} has equal dimensions $c$ by duplicating columns in $\BB_{bd}^\top$ (which does not change its image).
\begin{align*}
  \BB_3^\top \BB_3 = \diag(\JJ, \ldots, \JJ) \defeq c \PPi
\end{align*}
is a multiple of projection matrix where $\JJ$ is the all-one matrix in dimensions $c \times c$.
Since $ \BB_3^\top \BB_3 \yy \perp \ker(\BB_{4})$, we know $ \BB_3^\top \BB_3 \yy = c \PPi \yy \in \im(\BB_{4}^\top)$.
Since $\PPi$ is an orthogonal projection, we have 
$\norm{\PPi\yy} \le \norm{\PPi_{\im(\BB_{4}^\top)} \yy}$.
Note $\BB_{4}^\top \BB_{4}$ can be treated as a graph Laplacian matrix. By the eigenvalue bound in Theorem \ref{thm:diameter},
\begin{align*}
  \norm{\BB_{4} \yy}^2 \ge \Omega(r^{-1}) \norm{\PPi_{\im(\BB_{4}^\top)} \yy}^2
  \ge \Omega(r^{-1}) \norm{\PPi \yy}^2
  =  \Omega(r^{-1}) \cdot \frac{1}{c} \norm{\BB_3 \yy}^2.
\end{align*}
Therefore, 
\[
  \frac{\norm{\BB_3\yy}^2}{\norm{\BB_{4} \yy}^2} = O(r).\qedhere
\]
\end{proof}

Combining all the lemmas and claims above, we prove Lemma \ref{lem:up_proj_schur_precond}.

\subsection{Proof of Lemma \ref{lem:up_projection_approx}}

Given a vector $\bb$, we approximate $\Pup_1 \bb$ by the following steps: 
(1) compute $\bb_1 = \PPi_{\im(\partial_2[:,F])} \bb$;
(2) compute $\bb_2 = \partial_2^\top [C,:] \PkerF \bb$;
(3) approximately solve $\schur[\Ldown_2]_C \bb_3 = \bb_2$ via \pcg \ and get an approximate solution $\widetilde{\bb_3}$ up to error $\delta \le \frac{\eps}{\norm{\Lup_1}}$;
(4) compute $\bb_4 = \PkerF \partial_2 [:,C] \widetilde{\bb_3}$;
(5) compute $\bb_5 = \bb_1 + \bb_4$.
Let $\aPup_1$ be the above operator so that $\bb_5 = \aPup_1 \bb$.

\begin{claim}
$\aPup_1 \bb$ is in the image of $\Lup_1$.
\end{claim}
\begin{proof}
By Lemma \ref{lem:many_balls_exact_up_proj_restate}, the image of $\Lup_1$ can be decomposed as a direct sum of two orthogonal subspace: $\im(\partial_2[:,F])$ and $\im(\PkerF \partial_2 [:,C])$.
Note $\aPup_1 \bb =  \bb_1 + \bb_4$, where $\bb_1 \in \im(\partial_2[:,F])$ and $\bb_4 \in \im(\PkerF \partial_2 [:,C])$. So, $\aPup_1 \bb \in \im(\Lup_1)$.
\end{proof}

We bound the error of $\aPup_1 \bb$:
\begin{align*}
\norm{\aPup_1 \bb - \Pup_1 \bb}
& = \norm{\PkerF \partial_2[:,C] \left( (\schur[\Ldown_2]_C)^{\dagger} - \SS^{\dagger}
\right) \partial_2^\top [C,:] \PkerF \bb} \\
& \le \delta \norm{\partial_2^\top [C,:] \PkerF \bb}
\norm{\partial_2^\top [C,:] \PkerF} \\
& \le \delta \lambda_{\max}(\Lup_1) \norm{\Pup_1 \bb} \\
& \le \eps \norm{\Pup_1 \bb} 
\tag{by the setting of $\delta$}
\end{align*}

Finally, we bound the runtime of our approximate projection algorithm. The proof is similar to that of Lemma \ref{lem:solve_schur}.
By Theorem \ref{thm:PCG} and Lemma \ref{lem:up_proj_precondition_whole}, the number of \pcg \ (PCG) iterations is 
\[
\Otil\left( \kappa(\schur[\Ldown_2]_C, \partial_2^\top [C,:] \partial_2 [:,C])^{1/2}
\right) = \Otil(r^{1/2}).
\]
In each PCG iteration, we solve a system in $\partial_2^\top [C,:] \partial_2 [:,C]$ and multiply $\schur[\Ldown_2]_C$ with $O(1)$ vectors.
Recall 
\[
\schur[\Ldown_2]_C = \partial_2^\top [C,:] \PkerF \partial_2[:,C].
\]
In our preprocessing, we compute a Cholesky factorization of $\partial_2^\top [F,:] \PkerF \partial_2[:,F]$ in time $O(nr)$ (by the proof of Claim \ref{clm:up_proj_solve}); then, we can multiply $\schur[\Ldown_2]_C$ onto a vector in time $O(nr^{1/3})$.
Similarly, we solve a system in $\partial_2^\top [C,:] \partial_2 [:,C]$ by nested dissection. 
By our construction, $C$ has $O(r^{2/3} \cdot \frac{n}{r}) = O(nr^{-1/3})$ triangles.
The Cholesky factorization runs in time $O(n^2 r^{-2/3})$, and each system solve runs in time $O(n^{4/3} r^{-4/9})$.
So, the runtime per each PCG iteration is 
\[
O\left( nr^{1/3} + n^{4/3} r^{-4/9}
\right).
\]
Therefore, the total time is 
\begin{align*}
\Otil \left( r^{1/2} \left( nr^{1/3} + n^{4/3} r^{-4/9} \right)
+ nr + n^2 r^{-2/3}
\right)
= \Otil \left( nr + n^{4/3}r^{5/18} + n^2 r^{-2/3}
\right).
\end{align*}
We finish the proof of Lemma \ref{lem:up_projection_approx}.

\section{Proof of Theorem \ref{thm:main_ball}}
\label{sect:together}

Given all the four operators in Lemma \ref{lem:cohen_project}, \ref{lem:down_lap_solver}, \ref{lem:up_solver}, and \ref{lem:up_projection_approx}, we prove Theorem \ref{thm:main_ball}.

\begin{proof}[Proof of Theorem \ref{thm:main_ball}]
Let $\kappa$ be the maximum of $\kappa(\Ldown_1)$ and $\kappa(\Lup_1)$.
Let $\delta > 0$ be a parameter to be determined later. 
Let $\aPdown_1 = \aPdown_1( \delta ), \aPup_1 = \aPup_1(\delta )$ be defined in Lemma \ref{lem:cohen_project} and \ref{lem:up_projection_approx}, and let $\Zdown_1$ be the operator in Lemma \ref{lem:down_lap_solver} with no error and $\Zup_1$ in Lemma \ref{lem:up_solver} with error $\delta$.
Let 
\begin{align*}
& \abup \defeq \aPup_1 \bb, ~ \abdown \defeq \aPdown_1 \bb, \\
& \axup \defeq \Zup_1 \abup, ~
\axdown \defeq \Zdown_1 \abdown, \\
& \tilde{\xx} \defeq \aPup_1 \axup + \aPdown_1 \axdown.
\end{align*}
Then,
\begin{multline*}
\norm{\LL_1 \tilde{\xx} - \PPi_1 \bb}_2 \\
 \le \norm{\Lup_1 \aPup_1 \axup - \abup}_2
+ \norm{\Ldown_1 \aPdown_1 \axdown - \abdown}_2 + \norm{ \abup + \abdown - \PPi_1 \bb}_2.
\end{multline*}
We will upper bound the three terms on the right-hand side separately.
\begin{itemize}
    \item For the first term,
\begin{align*}
\norm{\Lup_1 \aPup_1 \axup - \abup}_2
& \le \norm{\Lup_1 \aPup_1 \axup - \Lup_1 \Pup_1 \axup}_2
+ \norm{\Lup_1 \axup - \abup}_2 \\
& \le \norm{\Lup_1 \aPup_1 \axup - \Lup_1 \Pup_1 \axup}_2 + \delta \norm{\abup}
\tag{by Lemma \ref{lem:up_solver}}
\end{align*}
By Lemma \ref{lem:up_projection_approx},
\begin{align*}
\norm{\Lup_1 \aPup_1 \axup - \Lup_1 \Pup_1 \axup}_2  
& \le \norm{\Lup_1}_2 \norm{(\aPup_1 - \Pup_1)\Pup_1 \axup}_2 \\
& \le \delta \norm{\Lup_1}_2  \norm{\Pup_1 \axup}_2.
\end{align*}
Let $\yy \defeq (\Lup_1)^{\dagger} \abup$.
By Lemma \ref{lem:up_solver},
\begin{align*}
\norm{\Pup_1 \axup - \yy}_2
\le \norm{(\Lup_1)^{\dagger}}_2 \norm{\Lup_1 \Pup_1 \axup - \abup}_2 
\le \delta \norm{(\Lup_1)^{\dagger}}_2 \norm{\abup}_2.
\end{align*}
By the triangle inequality,
\begin{align*}
\norm{\Pup_1 \axup }_2 
\le \norm{\yy}_2 + \delta \norm{(\Lup_1)^{\dagger}}_2 \norm{\abup}_2
\le (1+\delta) \norm{(\Lup_1)^{\dagger}}_2 \norm{\abup}_2.
\end{align*}
So, 
\[
\norm{\Lup_1 \aPup_1 \axup - \Lup_1 \Pup_1 \axup}_2    
\le \delta (1+\delta) \kappa \norm{\abup}_2,
\]
and 
\[
\norm{\Lup_1 \aPup_1 \axup - \abup}_2
\le 3\delta \kappa \norm{\abup}_2.
\]

\item For the second term, the operator $\Zdown_1$ has no error, which means $\Ldown_1 \axdown = \abdown$. Then,
\begin{align*}
 \norm{\Ldown_1 \aPdown_1 \axdown - \abdown}_2
 & = \norm{\Ldown_1 \aPdown_1 \axdown - \Ldown_1 \axdown}_2\\
 & \le \delta(1+\delta) \kappa \norm{\abdown}_2.
\end{align*}

\item For the third term,
\begin{align*}
\norm{\abup + \abdown - \PPi_1 \bb}_2^2
& = \norm{(\aPup - \Pup) \bb}_2^2 + \norm{(\aPdown - \Pdown) \bb}_2^2 \\
& \le \delta^2 \left( \norm{\Pup \bb}_2^2 + \norm{\Pdown \bb}_2^2
\right) 
\tag{by Lemma \ref{lem:cohen_project}, \ref{lem:up_projection_approx}, Fact \ref{fact:solvers}} \\
& = \delta^2 \norm{\PPi_1 \bb}_2^2.
\end{align*}
\end{itemize}

Combining all the above inequalities,
\begin{align*}
 \norm{\LL_1 \tilde{\xx} - \PPi_1 \bb}_2
& \le 3\delta \kappa \norm{\abup}_2
+ 2\delta \kappa \norm{\abdown}_2 
+ \delta \norm{\PPi_1 \bb}_2 \\
& \le 3\delta \kappa (1+\delta) \norm{\Pup_1 \bb}_2 
+ 2\delta \kappa (1+\delta) \norm{\Pdown_1 \bb}_2 + \delta \norm{\PPi_1 \bb}_2 \\
& \le 11 \delta \kappa \norm{\PPi_1 \bb}_2.
\end{align*}
Choosing $\delta \le \frac{\eps}{11 \kappa} $, we have 
\[
\norm{\LL_1 \tilde{\xx} - \PPi_1 \bb}_2
\le \eps \norm{\PPi_1 \bb}_2.
\]

The runtime bound follows Lemma \ref{lem:cohen_project}, \ref{lem:down_lap_solver}, \ref{lem:up_solver}, and \ref{lem:up_projection_approx}.
\end{proof}

%% file: r_hollowing_new.tex
%!TEX root = main.tex

\def\CH{\text{CH}}
\def\vol{\text{vol}}
\def\calH{\mathcal{H}}

\section{Computing an \texorpdfstring{$r$}{Lg}-Hollowing}

\label{sect:hollowing}

In this section, we prove Lemma \ref{lem:rhollowing} by describing a nearly-linear time algorithm (Algorithm \ref{alg:hollow}) that finds an $r$-hollowing of a pure $3$-complex $\calK$ embedded in $\R^3$ with $n$ stable simplexes each of volume $\Theta(1)$; in addition, $\calK$ has a well-shaped boundary structure with parameter $r=o(n)$ defined in Definition \ref{def:hole_requirement}.
Together with the algorithm that achieves Theorem \ref{thm:main_ball}, we prove Theorem \ref{thm:main_ball_rhollowing}.

We describe the ideas of Algorithm \ref{alg:hollow}.
Let $\mathbb{K}$ be the convex hull of the underlying topological space of $\calK$.
First, Algorithm \ref{alg:hollow} finds a \emph{nice bounding box} -- a box that encompasses $\mathbb{K}$ while ensuring its volume and aspect ratio are within constant factors of those of $\mathbb{K}$.
Lemma \ref{lem:nice_bounding_box} provides a linear time algorithm for finding a nice bounding box for $\mathbb{K}$ when the aspect ratio of $\mathbb{K}$ is $O(1)$.
Second, Algorithm \ref{alg:hollow} ``cuts'' the bounding box into $O(n/r)$ smaller boxes of equal volume using $2$-dimensional planes and turns these cutting planes into an $r$-hollowing.
Figure \ref{fig:complex_and_r_hollowing} visually illustrates the process of finding an $r$-hollowing.

We need the following lemma from \cite{BH01} to construct a nice bounding box.

\begin{lemma}[Lemma 3.4 of \cite{BH01}]
  \label{lm:bounding_box}
  Given a set $X$ of points in $\mathbb{R}^3$, we can compute in linear time a bounding box \footnote{A box that encompasses $X$.} $\calB$ with $\vol(\calB) \le 6 \sqrt{6} \vol(\calB^*)$, where $\vol(\cdot)$ is the volume and
  $\calB^*$ is a bounding box of $X$ with the minimum volume.
\end{lemma}

\begin{corollary}[Nice bounding box]
Let $\mathbb{K} \subset \R^3$ be a convex body with aspect ratio $O(1)$.
We can compute a nice bounding box of $\mathbb{K}$ in linear time.
\label{lem:nice_bounding_box}
\end{corollary}

\begin{proof}
Let $\calB$ be the bounding box computed by Lemma \ref{lm:bounding_box} with $X$ being the set of points on the boundary of $\mathbb{K}$.
We will show that $\calB$ is a nice bounding box of $\mathbb{K}$.
Let $\calO_1$ be the minimum-volume ball containing $\mathbb{K}$, and let $\calO_2$ be the maximum-volume ball contained in $\mathbb{K}$.
Since $\mathbb{K}$ has $O(1)$ aspect ratio, we know
\begin{align}
\vol(\calO_2) \le \vol(\mathbb{K}) \le \vol(\calO_1)
= O (\vol(\calO_2)).
\label{eqn:vol}
\end{align}
Let $\calB^*$ be a bounding box of $\mathbb{K}$ with the minimum volume, and let $\calB_1$ be a bounding box of $\calO_1$ with the minimum volume. Since $\calB_1$ contains $\mathbb{K}$ and $\calO_1$ is a Euclidean ball in $\R^3$, we have
\[
\vol(\calB^*) \le \vol(\calB_1) = O(\vol(\calO_1)).
\]
By Lemma \ref{lm:bounding_box} and Equation \eqref{eqn:vol}, we have 
\[
\vol(\calB) \le 6\sqrt{6} \vol(\calB^*)
= O(\vol(\calO_1)) = O(\vol(\calO_2)) = O(\vol(\mathbb{K})).
\]
The aspect ratio of the box $\calB$ is $O(1)$ times the ratio between the maximum and the minimum lengths of its edges.
Since $\calB$ contains $\calO_2$ and $\vol(\calB) = O (\vol(\calO_2))$, 
the aspect ratio of $\calB$ is $O(1)$.
\end{proof}

\begin{algorithm}
    \caption{\textsc{Hollowing}($\calK, r$)}
    \label{alg:hollow}
    \KwData{A pure $3$-complex $\calK$ embedded in $\R^3$ with $n$ stable simplexes each of volume $\Theta(1)$ and a well-shaped boundary structure with parameter $r = o(n)$}
    \KwResult{An $r$-hollowing $\calT$}

Find the convex hull of $\calK$, denoted by $\mathbb{K}$. Find a nice bounding box $\calB$ for $\mathbb{K}$.

For each pair of parallel faces of $\calB$, find $\lfloor n^{1/3}r^{-1/3} \rfloor$ evenly-spaced $2$-dimensional planes parallel to the face which divide $\calB$ into equal-volume pieces. We can slightly perturb the planes so that no plane intersects with any vertex of $\calK$ (see Figure \ref{fig:complex_and_r_hollowing}(b)). \label{lin:planes}

$\calT \leftarrow $ all the tetrahedrons in $\calK$ within $O(1)$ distance to the largest boundary of $\calK$ that form a spherical shell. \label{lin:hollow_initT}

  \For{each $2$-dimensional plane $P$\label{lin:for}}{
    $\calQ \leftarrow $ all the tetrahedrons of $\calK$ that intersect $P$.
 
    \If{$\calQ$ is not connected (i.e., $P$ intersects some boundary components inside $\mathbb{K}$)}{
    $\calQ \leftarrow \calQ ~\cup~ \bigcup_{\calH \in \text{$P$ intersected boundary components inside $\mathbb{K}$}}$ all the tetrahedrons within $O(1)$ distance to $\calH$ that are on the same side of $P$ and form a half spherical shell (see Figure \ref{fig:complex_and_r_hollowing}(c)).
    } 
    
    $\calT \leftarrow \calT ~ \cup \calQ$.
    \label{lin:plane2}
  }

  Expand $\calT$ such that its width reaches $5$.
  
    \Return{$\calT$}
\end{algorithm}

\begin{figure}[!ht] 
    \centering 
    \includegraphics[width=0.8\textwidth]{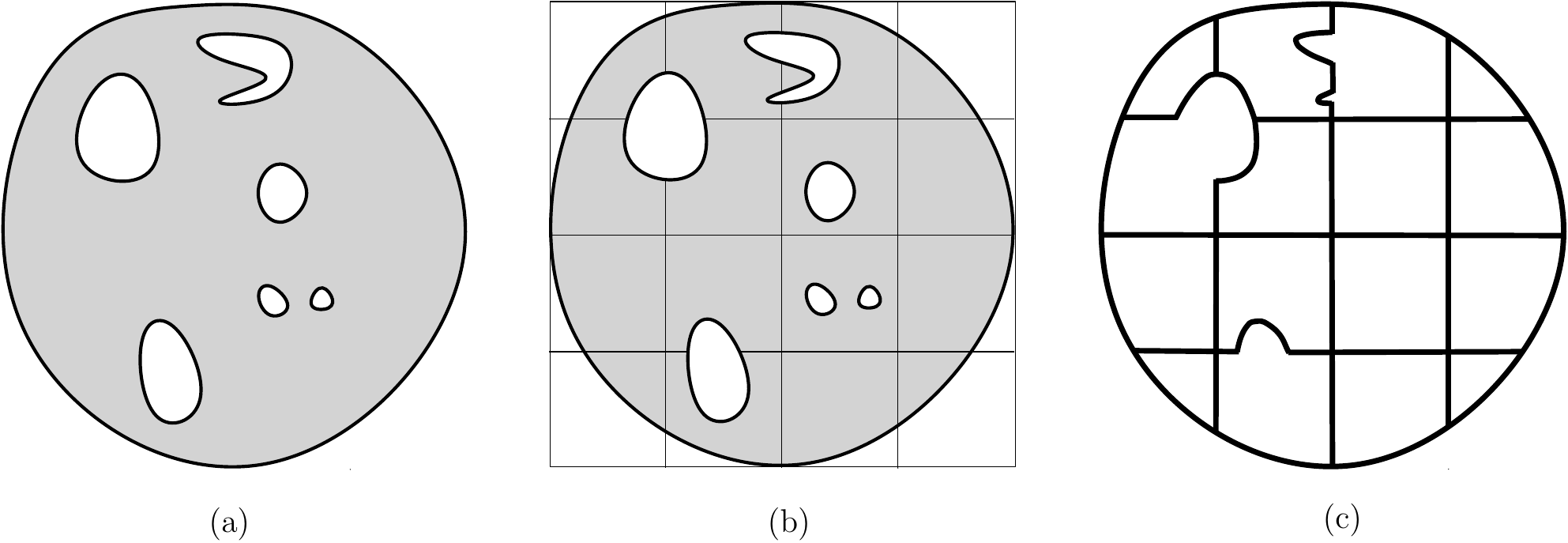} 
    \caption{(a) A crosse-sectional illustration of a 3-complex $\calK$ with several holes inside; (b) A nice bounding box of $\mathbb{K}$ with $\lfloor n^{1/3}r^{-1/3} \rfloor$ evenly-spaced $2$-dimensional planes; (c) An $r$-hollowing $\calT$ generated by Algorithm \ref{alg:hollow} consisting of simplexes that are ``close'' to the two-dimensional planes and parts of the boundary components of the intersecting holes inside with the planes.}
    \label{fig:complex_and_r_hollowing} 
\end{figure}

\begin{proof}[Proof of Lemma \ref{lem:rhollowing}]
The convex hull of $\calK$ can be computed in $O(n \log n)$ time by \cite{chan96}, and the nice bounding box can be computed in $O(n)$ time by Corollary \ref{lem:nice_bounding_box}.
The remaining steps of the algorithm can be implemented in $O(n)$ time. In the rest of the proof, we will show $\calT$ returned by Algorithm \ref{alg:hollow} is an $r$-hollowing of $\calK$.

By Assumption \ref{asp:1}, \ref{asp:2} and \ref{asp:3} and the assumption that each simplex has $\Theta(1)$ volume, the volume of the convex hull $\mathbb{K}$ is $\Theta(n)$; the maximum volume is attained when $\calK$ has $\Theta(n/r)$ boundary components and each corresponds to a ``hole'' of volume $\Theta(r)$.
By Lemma \ref{lem:nice_bounding_box}, we have $\vol(\calB) = \Theta(\mathbb{K}) = \Theta(n)$.
In Algorithm \ref{alg:hollow}, the $2$-dimensional planes divide the box $\calB$ into $O(n/r)$ smaller boxes each of volume $O(r)$ and surface area $O(r^{2/3})$. 
By our construction of $\calT$, each smaller box corresponds to a region; thus, there are $O(n/r)$ regions. 
By Assumption \ref{asp:3}, each region of $\calT$ has $O(r)$ simplexes and $O(r^{2/3})$ boundary simplexes. 
Moreover, the boundary of each region triangulates a spherical shell in $\R^3$ by construction. 
Additionally, the diameter of the underlying topological space of each region is upper bounded by the triangle diameter of the small box plus $\Theta(1)$ times 
the $1$-skeleton diameter of boundary components.
By Assumption \ref{asp:2}, each region has diameter $O(r^{1/3})$.

To conclude, $\calT$ satisfies all the conditions in Definition \ref{def:hollow} and is an $r$-hollowing of $\calK$.
\end{proof}

%% file: another_precond_solver.tex
\section{A Faster Solver for Up-Laplacian}
\label{sect:solver_bdry}

In this section, we introduce a faster solver for linear equations in $\Lup_1$, which approximately solves $\Lup_1 \xx = \bb$ given $\bb$ lies in the image of $\Lup_1$. 
While this solver does not improve the runtime of Theorem \ref{thm:main_ball} (as the primary bottleneck remains the approximate up-projection operator in Lemma \ref{lem:up_projection_approx}), we believe it holds independent interest.
When the first Betti number of the pure $3$-complex $\calK$ is zero, one can approximately project a vector onto the image of $\Lup_1$ in nearly-linear time (as stated in Lemma \ref{lem:cohen_project_up}). In such cases, incorporating the faster solver for linear equations in $\Lup_1$ improves the overall runtime.

\begin{lemma}[Up-projection operator, Lemma 3.2 of \cite{CMFNPW14}]
Let $\calK$ be a $3$-complex with $n$ simplexes.
If the first Betti number of $\calK$ is $0$, there exists an operator $\aPup_1$ such that 
\begin{align*}
    & (1-\eps)\Pup_1 \pleq \aPup_1(\eps) \pleq \Pup_1.
\end{align*}
In addition, for any $\bb \in \mathbb{R}^{n}$, we can compute $\aPup_1 (\eps) \cdot \bb$ in time $\Otil(n \log(1/\eps))$. 
\label{lem:cohen_project_up}
\end{lemma}

The enhanced performance results from the reduced size of $r$-hollowings needed for preconditioning the Schur complement of $\Lup_1$. Specifically, we will use a triangulated sphere, instead of a spherical shell, as a region boundary.
This reduced version of $r$-hollowing is referred to as the \emph{sphere $r$-hollowing} and is defined as follows.

\begin{definition}[Sphere $r$-hollowing]
Let $\calK$ be a $3$-complex with $n$ simplexes, and let $r = o(n)$ be a positive integer.
We divide $\calK$ into $O(n/r)$ regions each of $O(r)$ simplexes and $O(r^{2/3})$ boundary simplexes. Only boundary simplexes can belong to more than one region.
The boundary of each region triangulates a 2D sphere and has triangle diameter $O(r^{1/3})$.
In addition, each region intersects with at most $O(1)$ other regions. If two regions intersect, their intersection is either a triangulated disc or a path on the boundary of the two regions.
The union of all boundary simplexes of each region is referred to as a sphere $r$-hollowing of $\calK$.
\label{def:surface_hollow}
\end{definition}

We remark that if $\calK$ has a well-shaped boundary structure defined in Definition \ref{def:hole_requirement},
we can also find a sphere $r$-hollowing in nearly-linear time with a slight modification to Algorithm \ref{alg:hollow}.

We state the Lemma of the faster up-Laplacian solver with a sphere $r$-hollowing as follows, where we set $r = O(\sqrt{n})$.

\begin{lemma}[Faster up-Laplacian solver with a sphere $r$-hollowing]
Let $\calK$ be a pure $3$-complex embedded in $\R^3$ and composed of $n$ stable simplexes. 
Let $\calT$ be a sphere $r$-hollowing of $\calK$. 
Then for any $\eps > 0$, there exists an operator $\Zup_1$ such that
\[
\forall \bb \in \im(\Lup_1), ~ \norm{\Lup_1 \Zup_1 \bb - \bb}_2 \le \eps \norm{\bb}_2.
\]
In addition, $\Zup_1 \bb$ can be computed in time 
$\Otil \left( nr + n^2 r^{-4/3} + n^3 r^{-3} 
\right)$. In particular, when $r = \Theta(\sqrt{n})$, the runtime is $\Otil \left( n^{3/2}\right)$. 
\label{lem:up_solver_surface}
\end{lemma}

Compared to the runtime of $\Otil(n^{8/5})$ in Lemma \ref{lem:up_solver} when $r = \Theta(\sqrt{n})$, the faster up-Laplacian solver achieves an acceleration of $\Otil(n^{1/10})$. 
The solver's outline aligns with Algorithm \ref{alg:solver}, albeit employing a different form of $r$-hollowing. 
The relative condition number of the Schur complement and the sphere $r$-hollowing boundary is $O(r)$ (whose proof is similar to Section \ref{sect:condNum} and we omit it).
The improved runtime stems from a faster solver designed specifically for the preconditioner $\Lup_{1,\calT}$ operating on the sphere $r$-hollowing of $\calT$. 
Further details regarding this accelerated solver will be presented in the rest of the section.

\begin{lemma}[Solver for the surface preconditioner]
Let $\calK$ be a pure $3$-complex embedded in $\R^3$ and composed of $n$ stable simplexes. 
Let $\calT$ be a sphere $r$-hollowing of $\calK$. 
With a pre-processing in time $O \left( n^2 r^{-4/3} + n^3 r^{-3} \right)$,
given any vector $\bb \in \im(\Lup_{1})$, we can find an $\xx$ such that $\Lup_{1,\calT} \xx = \bb$ in time $O \left( n r^{-1/3} + n^2 r^{-2} \right)$.
\label{lem:solveS}
\end{lemma}

We describe our high-level idea for solving a system in $\Lup_{1,\calT}$.
Given our definition of sphere $r$-hollowing, 
$\calT$ is a union of discs where two discs only share boundary edges.
We notice that the first up-Laplacian operator applied to the interior of a disc can be treated as a first down-Laplacian operator. This observation allows us to solve linear equations in the interior of a disc in linear time using the down-Laplacian solver described in Lemma \ref{lem:down_lap_solver}.
With this in mind, we apply Lemma \ref{lem:combine_solvers} to $\Lup_{1,\calT}$ with $F$ being all the interior edges of the discs and $C$ the remaining edges. 
To further reduce the runtime of solving the Schur complement onto $C$, we observe that many of the boundary edges of the discs are linearly dependent.
Consequently, we can eliminate these linearly dependent edges and reduce the system to a smaller, more manageable one.

\subsection{Reducing to a Smaller System}

We partition the edges in $\calT$ into $E_1 \cup E_2$, where $E_1$ contains all the interior edges of the discs and $E_2$ the boundary edges.
For each $e \in E_2$, let $D_e$ be the set of \emph{discs} that contain $e$.

\begin{claim}
Let $R$ be an arbitrary region of $\calK$. 
Among all the boundary edges in $E_2$ of $R$, there are $O(1)$ distinct values for all $D_e$.
\end{claim}
\begin{proof}
By our definition of sphere $r$-hollowing in Definition \ref{def:surface_hollow}, $R$ intersects with $O(1)$ regions; each intersection is a disc on the boundary of $R$, and two discs only intersect on their boundaries. So, the boundary of $R$ is divided into $O(1)$ discs. This also holds for any other region.
Let $\calD_1, \ldots, \calD_l$, where $l = O(1)$, be all the discs belonging to the regions that intersect with $R$.
For any boundary edge $e$ of $R$, $D_e$ is a subset of $\{\calD_1, \ldots, \calD_l\}$. Since $l = O(1)$, there are $O(1)$ subsets of $\{\calD_1, \ldots, \calD_l\}$.
\end{proof}

Let $\BB_1 = \partial_{2,\calT}[E_1, :]$ and $\BB_2 = \partial_{2,\calT}[E_2, :]$. 
Without loss of generality (by Claim \ref{clm:orient_triangle}), we assume all the triangles in the same disc of $\calT$ have the same orientation. Under this assumption, each row of $\BB_1$ has exactly one entry being $1$ and one being $-1$ and all the others $0$.

\begin{claim}
Let $e_1, e_2 \in E_2$ such that $D_{e_1} = D_{e_2}$.
Then, the row $\BB_2[e_2,:]$ is a linear combination of the rows in $\BB_1$ and $\BB_2[e_1,:]$. 
\label{clm:linear_comb}
\end{claim}

\begin{proof}
Suppose $D_{e_1} = D_{e_2} = \{\calD_1, \ldots, \calD_l \}$ where $l = O(1)$.
Suppose $e_1$ belongs to triangles $\triangle_{1} \in \calD_{1}, \ldots, \triangle_{l} \in \calD_{l}$  and $e_2$ belongs to $\triangle_{1}' \in \calD_{1}, \ldots, \triangle_{l}' \in \calD_{l}$.
Since all the triangles in the same disc of $\calT$ have the same orientation, there exists $s \in \{\pm 1\}$ such that 
\[
\BB_2[e_1, \triangle_{1}] = s \BB_2[e_2, \triangle_{1}'],
\ldots,
\BB_2[e_1, \triangle_{l}] = s \BB_2[e_2, \triangle_{l}'].
\]
For any $1 \le i \le l$, let $T_{i}$ be the set of triangles in $\calD_{i}$ and $E_{1,i}$ be the set of interior edges in $\calD_{i}$. We claim there exists a vector $\aa_i$ such that 
\begin{align}
\BB_2[e_2, T_i] = s\BB_2[e_1, T_i] + \aa_i^\top \BB_1[E_{1,i},T_i].
\label{eqn:path_goal}
\end{align}
Since $\BB_2[e_2,:] = \begin{pmatrix}
\BB_2[e_2, T_1] & \cdots & \BB_2[e_2, T_l] & {\bf 0}
\end{pmatrix}$ and $E_{1,1}, \ldots, E_{1,l}$ are disjoint, Equation \eqref{eqn:path_goal} implies  there exists a vector $\aa$ such that 
\[
\BB_2[e_2, :] = s\BB_2[e_1, :] + \aa^\top \BB_1.
\]

It remains to prove Equation \eqref{eqn:path_goal}.
Since both $\triangle_i, \triangle_i' \in \calD_i$, we can find a sequence of triangles in $\calD_{i}$: $\widehat{\triangle}_0 = \triangle_{i}, \widehat{\triangle}_{1}, \widehat{\triangle}_2, \ldots,  \widehat{\triangle}_j = \triangle_{i}'$ such that 
for any $1 \le h \le j$, the two triangles $\widehat{\triangle}_{h-1}$ and $\widehat{\triangle}_{h}$ share an interior edge of $\calD_{l}$, denoted by $\widehat{e}_{h}$.
For every $1 \le h \le j$, the row $\BB_1[\widehat{e}_{h}, T_{i}]$ has exactly one entry being $1$ and one entry $-1$ indexed at $\widehat{\triangle}_{h-1}, \widehat{\triangle}_{h}$ and all others $0$. 
Denote $t = s\BB_2[e_1, \triangle_i] \in \{\pm 1\}$.
Thus, there exists $s_1 \in \{\pm 1\}$ such that $s \BB_2[e_1, T_i] + s_1 \BB_1[\widehat{e}_1, T_i] = \begin{pmatrix}
0 & \cdots & 0 & t & 0 & \cdots & 0
\end{pmatrix} $ where the entry $t$ has index $\widehat{\triangle}_1$.
By induction on $h$, there exist $s_1, \ldots, s_j \in \{\pm 1 \}$ such that 
\begin{align*}
s \BB_2[e_1, T_i] + s_1 \BB_1[\widehat{e}_1, T_i] + \cdots + s_j \BB_1[\widehat{e}_j, T_i]
= \begin{pmatrix}
0 & \cdots & 0 & t & 0 & \cdots & 0    
\end{pmatrix},
% = \BB_2[e_2, T_i].
\end{align*}
where the index of entry $t$ is $\triangle'_i$.
This proves Equation \eqref{eqn:path_goal}.
\end{proof}

Thanks to the linear dependence in Claim \ref{clm:linear_comb}, we can eliminate most rows in $\BB_2$.
Among the edges with the same value of $D_e$, we arbitrarily select one and include it in a set denoted as $\widehat{E}_2$. Then, we define 
\[
\BBhat_2 \defeq \BB_2[\widehat{E}_2, :],~ 
\BB \defeq \begin{pmatrix}
\BB_1 \\
\BBhat_2
\end{pmatrix}, \text{ and }
\LupOneT \defeq \BB \WW_{2,\calT} \BB^\top.
\]

\begin{claim}
Let $\bb \in \im(\partial_{2,\calT})$ and let $\bb_1 = \bb[E_1 \cup \widehat{E}_2]$. 
If $\xx_1$ satisfies $\LupOneT \xx_1 = \bb_1$, then $\xx = \begin{pmatrix}
    \xx_1 \\
    {\bf 0}
\end{pmatrix}$ satisfies $\Lup_{1,\calT} \xx = \bb $.
\label{clm:precondition_small_system}
\end{claim}

\begin{proof}
By Claim \ref{clm:linear_comb}, there exists a matrix $\HH$ such that $\BB_2[E_2 \setminus \widehat{E}_2, :] = \HH \BB$. Thus, $\partial_{2,\calT}$ can be written as $\begin{pmatrix}
    \BB \\
    \HH \BB
\end{pmatrix}$ up to row permutation.
Since $\bb \in \im(\partial_{2,\calT})$, we know $\bb_1 \in \im(\BB)$ and thus $\LupOneT \xx_1 = \bb_1$ is feasible and $\bb_2 = \HH \bb_1$.
Then,
\begin{align*}
\Lup_{1,\calT} \xx 
& = \begin{pmatrix}
     \BB \\
     \HH \BB
\end{pmatrix}      \WW_{2,\calT} 
\begin{pmatrix}
\BB^\top & \BB^\top \HH^\top
\end{pmatrix} \begin{pmatrix}
     \xx_1 \\
     \veczero
\end{pmatrix} 
= \begin{pmatrix}
     \BB \WW_{2,\calT} \BB^\top \xx_1 \\
     \HH \BB \WW_{2,\calT} \BB^\top \xx_1
\end{pmatrix} 
= \begin{pmatrix}
     \bb_1 \\
     \bb_2
\end{pmatrix}.
\end{align*}
\end{proof}

\subsubsection{Solver for the Smaller System}

Let $\MM_{11} \defeq \BB_1 \WW_{2,\calT} \BB_1^\top, \MM_{12} \defeq \BB_1 \WW_{2,\calT} \BBhat_2^\top, 
\MM_{22} \defeq \BBhat_2 \WW_{2,\calT} \BBhat_2^\top$.
Then, we can write $\MM$ as a block matrix:
\[
\MM = \begin{pmatrix}
\MM_{11} & \MM_{12} \\
\MM_{12}^\top & \MM_{22}
\end{pmatrix}.
\]
By Lemma \ref{lem:combine_solvers}, it suffices to design efficient solvers for systems in $\MM_{11}$ and systems in the Schur complement:
\begin{align}
    \schur[\LupOneT]_{\widehat{E}_2}
    = \MM_{22} - \MM_{12}^\top \MM_{11}^{\dagger} \MM_{12}.
     \label{eqn:schur1S}
\end{align}
Let $m_1 = \abs{E_1}, m_2 = \abs{\widehat{E}_2}$.

\begin{claim}
For any $\bb \in \im(\MM_{11})$, we can compute $\xx$ satisfying $\MM_{11} \xx = \bb$ in time $O(m_1)$.
\label{clm:solve_m11}
\end{claim}

\begin{proof}
Without loss of generality (by Claim \ref{clm:orient_triangle}), we assume all the triangles in the same disc of $\calT$ have the same orientation.
Each row of $\BB_1$ has exactly one entry with value $1$ and one with value $-1$ and all others $0$.
Thus, $\MM_{11} =  \BB_1 \WW_2 \BB_1^\top$ can be viewed as a first down-Laplacian operator.
By Lemma \ref{lem:down_lap_solver}, we can solve a system in $\MM_{11}$ in linear time.    
\end{proof}

\begin{claim}
We can compute $\left( \schur[\LupOneT]_{\widehat{E}_2} \right)^{\dagger}$ in time $O(m_1m_2 + m_2^3)$.
\label{clm:schur_inverse}
\end{claim}

\begin{proof}
We first compute the Schur complement $\schur[\LupOneT]_{\widehat{E}_2}$ defined in Equation \eqref{eqn:schur1S}.
Note that $\MM_{22} = \BBhat_2 \WW_{2,\calT} \BBhat_2^\top$ can be written as a weighted sum of the outer products of each nonzero column of $\BBhat_2$. Since $\BBhat_2$ has $O(m_2)$ nonzero columns and each nonzero column of $\BBhat_2$ has at most $3$ nonzero entries, we can compute $\MM_{22}$ in time $O(m_2)$.
Similarly, we can compute $\MM_{12} =  \BB_1 \WW_{2,\calT} \BBhat_2^\top$ in time $O(m_1 + m_2)$.
Then, by Lemma \ref{clm:solve_m11}, we can find a matrix $\XX \in \mathbb{R}^{m_1 \times m_2}$ such that $\MM_{11} \XX = \MM_{12}$ in time $O(m_1 m_2)$.
We can write $\XX = \XX_1 + \XX_2$ where $\XX_1 = \MM_{11}^{\dagger} \MM_{12}$ and each column of $\XX_2$ is orthogonal to $\im(\MM_{11}) = \im(\BB_1)$.
Then, 
\[
\MM_{12}^\top \XX = \BBhat_2 \WW_{2,\calT} \BB_1^\top \XX_1 = \MM_{12}^\top \MM_{11}^{\dagger} \MM_{12}.
\]
Multiplying $\MM_{12}^\top \XX$ runs in time $O(m_1m_2 + m_2^2)$ since $\BB_1$ has $O(m_1)$ nonzeros and $\widehat{\BB}_2$ has $O(m_2)$ nonzeros.
Thus, the total time of computing the Schur complement
is $O(m_1m_2 + m_2^2)$.

We then show how to compute the pseudo-inverse of the Schur complement in time $O(m_2^3)$.
We can decompose 
\[
 \schur[\LupOneT]_{\widehat{E}_2} = \QQ \begin{pmatrix}
    \TT & {\bf 0} \\
    {\bf 0} & {\bf 0}
\end{pmatrix} \ZZ,
\]
where $\QQ, \ZZ \in \mathbb{R}^{m_2 \times m_2}$ are orthogonal matrices and $\TT \in \mathbb{R}^{k \times k}$ is a full-rank matrix and $k = \rank( \schur[\LupOneT]_{\widehat{E}_2})$. Such a decomposition can be obtained in time $O(m_2^3)$ by Householder QR factorization with column pivoting (Chapter 5.4 of \cite{VG96}).
Then, 
\begin{align*}
\left(  \schur[\LupOneT]_{\widehat{E}_2} \right)^{\dagger} = \ZZ^\top \begin{pmatrix}
    \TT^{-1} & {\bf 0} \\
    {\bf 0} & {\bf 0}
\end{pmatrix} \QQ^\top.
\end{align*}
The inverse of $\TT$ can be computed in time $O(m_2^3)$ by matrix multiplication (Theorem 28.2 of \cite{CLRS22}). Thus, the total runtime is $O(m_1m_2 + m_2^3)$.
\end{proof}

Combining Claim \ref{clm:solve_m11} and \ref{clm:schur_inverse} and Lemma \ref{lem:combine_solvers} yields the following claim.

\begin{claim} 
With a pre-processing time $O(m_1m_2+ m_2^3)$, given any $\bb \in \im(\LupOneT)$, we can find $\xx$ satisfying $\LupOneT \xx = \bb$ in time $O(m_1 + m_2^2)$.
\label{lem:solveM}
\end{claim}

\begin{proof}
In the pre-processing, we compute the pseudo-inverse of $\schur[\LupOneT]_{\widehat{E}_2}$.
Then, we can solve the system in the Schur complement by multiplying $\left(\schur[\LupOneT]_{\widehat{E}_2}
\right)^{\dagger}$ with the right-hand side vector.
\end{proof}

We are ready to prove Lemma \ref{lem:solveS}, which bounds the runtime of solving a system in $\Lup_{1,\calT}$.

\begin{proof}[Proof of Lemma \ref{lem:solveS}]
Recall that $\calT$ is a sphere $r$-hollowing of $\calK$ with $n$ simplexes. So,
\[
m_1 = O \left( \frac{n}{r} \cdot r^{2/3} \right) = O(n r^{-1/3}), 
~ m_2 = O(n r^{-1}).
\]
By Claim \ref{clm:precondition_small_system}, it suffices to solve the system in $\MM$. By Claim \ref{lem:solveM}, we can solve a system in $\MM$ in time 
$O \left( n r^{-1/3} + n^2 r^{-2} \right)$ with a pre-processing time 
$O \left( n^2 r^{-4/3} + n^3 r^{-3} \right)$.
\end{proof}

\begin{proof}[Proof of Lemma \ref{lem:up_solver_surface}]
The pre-processing time is $O(nr + n^2 r^{-4/3} + n^3 r^{-3})$.
The runtime of solving the ``$F$" part is $O(nr^{1/3})$. The runtime of solving the system in the Schur complement is $\Otil(\sqrt{r}(n r^{-1/3} + n^2 r^{-2}))$. By Lemma \ref{lem:combine_solvers}, the total runtime is 
\begin{align*}
\Otil \left( nr + n^2 r^{-4/3} + n^3 r^{-3} 
\right).
\end{align*}
When $r = \Theta(\sqrt{n})$, the above runtime is $\Otil(n^{3/2})$.
\end{proof}

%% file: multiple_balls.tex
%!TEX root = main.tex

\section{Union of Pure \texorpdfstring{$3$}{Lg}-Complexes}
\label{sect:many_balls}

In this section, we consider a union $\calU$ of $h$ pure $3$-complex chunks $\calK_1, \ldots, \calK_h$ that are composed of stable simplexes, and these chunks are glued together by identifying certain subsets of their exterior simplexes. 
Let $k$ be the number of simplexes shared by more than one chunk.
We remark that $\calU$ may \emph{not} be embeddable in $\R^3$ and the first and second Betti numbers of $\calU$ are \emph{no longer zero}.
This makes designing efficient solvers for $1$-Laplacian systems of $\calU$ much harder.
Edelsbrunner and Parsa \cite{EP14} showed that computing the first Betti number of a
simplicial complex linearly embedded in $\R^4$ with $m$ simplexes is as hard as computing the rank of a $0$-$1$ matrix with $\Theta(m)$ nonzeros; 
Ding, Kyng, Probst Gutenberg, and Zhang \cite{DKPZ22} showed that (approximately) solving $1$-Laplacian systems for simplicial complexes in $\R^4$ is as hard as (approximately) solving general sparse systems of linear equations.

We design a $1$-Laplacian solver for $\calU$ whose runtime is comparable to that of the $1$-Laplacians solver for a single pure $3$-complex when both $h$ and $k$ are small and prove Theorem \ref{thm:main}.
The approximate down-projection operator for in Lemma \ref{lem:cohen_project}
and the approximate solver for systems in the down-Laplacian in Lemma \ref{lem:down_lap_solver} hold for any simplicial complexes. 
Thus, it suffices to generalize the approximate up-projection operator in Lemma \ref{lem:up_projection_approx} and the approximate solver for systems in the up-Laplacian in Lemma \ref{lem:up_solver} from a single chunk to a union of chunks.

\begin{lemma}[Up-Laplacian solver]
Let $\calU$ be a union of pure $3$-complexes glued together by identifying certain subsets of their exterior simplexes. Each $3$-complex chunk is embedded in $\R^3$, comprises $n_i$ stable simplexes, and has a known $\Theta (n_i^{3/5})$-hollowing.
For any $\bb \in \im(\Lup_1)$ and $\eps > 0$, we can compute an $\tilde{\xx}$ such that $\norm{\Lup_1 \tilde{\xx} - \bb}_2 \le \eps \norm{\bb}_2$ in time
$\Otil \left( n^{8/5} + n^{3/10} k^2 + k^3 \right)$ where $n$ is the number of simplexes in $\calU$ and $k$ is the number of simplexes shared by more than one chunk. 
\label{lem:many_balls_up_solver}
\end{lemma}

\begin{lemma}[Up-projection operator]
Let $\calU$ be a $3$-complex satisfying the conditions in Lemma \ref{lem:up_project}. 
For any $\eps > 0$, 
there exists an operator $\aPup_1$ such that
\[
\forall \bb, \norm{\aPup_1 \bb - \Pup_1 \bb}_2 \le \eps \norm{\Pup_1\bb}_2.
\]
In addition, $\aPup_1 \bb$ can be computed in time 
$\Otil \left( n^{8/5} + n^{3/10} k^2 + k^3 \right)$,
where $n$ is the number of exterior simplexes of $\calU$ and 
$k$ is the number of simplexes shared by more than one chunk. 
\label{lem:up_project}
\end{lemma}

Our approaches align with those for a single chunk.
We partition the simplexes of $\calK$ into $F \cup C$, then deal with the ``$F$" part and the Schur complement separately.
Given the definition of $r$-hollowing, the exterior simplexes of each chunk $\calK_i$ must belong to the boundary of some region.
We let $C$ be the union of the hollowing boundary of each chunk and let $F$ be the remaining simplexes.
Then, $F$ is a union of disjoint subcomplexes each embedded in $\R^3$ and can be handled by \nd.
We precondition the Schur complement by the union of the hollowing boundaries. 
However, systems in this preconditioner cannot be approximately solved \nd \ directly since it may not allow an embedding in $\R^3$. We will need a slightly more careful treatment.

\begin{proof}[Proof of Lemma \ref{lem:many_balls_up_solver}]

We let $C$ be the union of the hollowing boundary edges in each chunk and let $F$ be the union of the hollowing interior edges.
Suppose $\calU$ has $h$ pure $3$-complex chunks.
By our definition of $r$-hollowing, we can write $\Lup_1[F,F] = \diag(\Lup_1[F_1, F_1], \ldots, \Lup_1[F_h, F_h])$ where $F_i$ contains all the hollowing interior edges in the $i$th chunk.
Let $r_i = n_i^{3/5}$ be the hollowing parameter for the $i$th chunk.
By Lemma \ref{lem:solve_interior}, with a pre-processing time 
\[
O \left( \sum_{i=1}^h n_i r_i
\right) = 
O \left( \sum_{i=1}^h n_i^{8/5} \right) = O\left(n^{8/5} \right),
\]
for any $\bb \in \im(\Lup_1[F,F])$, we can find $\xx$ such that $\Lup_1[F,F] \xx = \bb$ in time 
\[
O\left( \sum_{i=1}^h n_i r_i^{1/3} \right) = 
O\left( \sum_{i=1}^h n_i^{6/5} \right)
= O \left( n^{6/5} \right).
\]

To solve the system in the Schur complement $\schur[\Lup_1]_C$, we precondition it by the union of the hollowing boundaries of each chunk, denoted by $\calT_U$.
By Claim \ref{clm:schur_decompose_approx} and Lemma \ref{lem:eigS_inner_spheres}, the relative condition number is $O(\max_i n_i^{3/5})$.
Let $C_1 \subset C$ contain the edges shared by more than one chunk and $C_2 = C \setminus C_1$.
Then, the submatrix $\Lup_{1,\calT_U}[C_2, C_2]$ is a block diagonal matrix where each block corresponds to a chunk.
We solve $\Lup_{1,\calT_U}$ by Lemma \ref{lem:combine_solvers}: We solve a system in $\Lup_{1,\calT_U}[C_2, C_2]$ by \nd \ and solve the Schur complement onto $C_1$ by directly inverting the Schur complement.
With a pre-processing time 
\[
O \left( \sum_{i=1}^h \left( \frac{n_i}{r_i} \cdot r_i^{2/3}
\right)^2 + k^3
\right) = 
O \left( \sum_{i=1}^h n_i^{8/5} + k^3
\right) = O(n^{8/5} + k^3),
\]
we can solve a system in $\Lup_{1,\calT_U}$ in time 
\[
O\left( \sum_{i=1}^h \left( \frac{n_i}{r_i} \cdot r_i^{2/3}
\right)^{4/3} + k^2 \right) = 
O\left( \sum_{i=1}^h n_i^{16/15} + k^2
\right) = O(n^{16/15} + k^2).
\]
Therefore, the total runtime of approximately solving a system in $\Lup_1$ is 
$\Otil \left( n^{8/5} + n^{3/10} k^2 + k^3 \right)$.
\end{proof}

We can prove Lemma \ref{lem:up_project} by a similar argument.

%% file: appendix.tex
%!TEX root = main.tex

\section{Missing Linear Algebra Proofs}

\subsection{Missing Proofs in Section \ref{sect:linear_algebra}}

\label{sect:proof_linear_algebra}

\begin{proof}[Proof of Fact \ref{fact:image}]
By the definition of image, $\im(\AA) \supseteq \im(\AA \AA^\top)$.
It suffices to show $\im(\AA) \subseteq \im(\AA \AA^\top)$.
Let $\xx$ be an arbitrary vector in $\im(\AA)$. Then, $\xx = \AA \yy$ for some $\yy \in \mathbb{R}^{n}$. Write $\yy = \yy_1 + \yy_2$ such that $\yy_1 \in \im(\AA^\top)$ and $\yy_2 \in \ker(\AA)$.
Let $\zz$ satisfy $\AA^\top \zz = \yy_1$. Then, 
\[
\AA \AA^\top \zz = \AA \yy_1 = \AA \yy = \xx.
\]
Thus, $\xx \in \im(\AA \AA^\top)$.
\end{proof}

\begin{proof}[Proof of Fact \ref{fact:identity}]
We multiply the right-hand side:
\begin{align*}
& \begin{pmatrix}
    \II & \\
    \AA[C,F] \AA[F, F]^{\dagger} & \II 
  \end{pmatrix} 
  \begin{pmatrix}
    \AA[F, F] & \\
    & \schur[\AA]_C
  \end{pmatrix}
  \begin{pmatrix}
    \II & \AA[F, F]^{\dagger} \AA[F, C] \\
    & \II
  \end{pmatrix} \\
= & \begin{pmatrix}
    \II & \\
    \AA[C,F] \AA[F, F]^{\dagger} & \II 
  \end{pmatrix}  \begin{pmatrix}
\AA[F,F] & \AA[F,C] \\
 & \schur[\AA]_C
  \end{pmatrix} \\
= & \begin{pmatrix}
\AA[F,F] & \AA[F,C] \\
\AA[C,F] & \AA[C,F] \AA[F,F]^{\dagger} \AA[F,C] + \schur[\AA]_C 
\end{pmatrix} \\
= & \begin{pmatrix}
\AA[F,F] & \AA[F,C] \\
\AA[C,F] & \AA[C,C] 
\end{pmatrix} = \AA.\qedhere
\end{align*}
\end{proof}

\begin{proof}[Proof of Fact \ref{fact:schur_project_kernel}]
By the definition of the Schur complement,
\begin{multline*}
\schur[\AA]_C
 = \BB_C \BB_C^\top - \BB_C \BB_F^\top \left(
\BB_F \BB_F^\top
\right)^{\dagger} \BB_F \BB_C^\top \\
 = \BB_C \left( \II - \BB_F^\top \left(
\BB_F \BB_F^\top
\right)^{\dagger} \BB_F \right) \BB_C^\top
= \BB_C \PPi_{\ker(\BB_F)} \BB_C^\top.
\end{multline*}
\end{proof}

\begin{proof}[Proof of Fact \ref{fact:solvers}]
For the first statement,
\begin{align*}
\norm{\AA \ZZ \bb - \bb}_2
& = \norm{\AA^{1/2} \left( \AA^{1/2} \ZZ \AA^{1/2} - \PPi_{\im(\AA)} \right) \AA^{\dagger /2} \bb}_2 \\
& \le \sqrt{\kappa(\AA)} \norm{  \AA^{1/2} \ZZ \AA^{1/2} - \PPi_{\im(\AA)} }_2 \norm{ \bb}_2.
\end{align*}
By the assumption, $(1-\eps) \PPi_{\im(\AA)} \pleq \AA^{1/2} \ZZ \AA^{1/2} \pleq (1+\eps) \PPi_{\im(\AA)} $. Thus, 
\[
\norm{\AA \ZZ \bb - \bb}_2 \le \eps \sqrt{\kappa(\AA)} \norm{\bb}_2.
\]

For the second statement, we know $\norm{\AA \PPi_{\im(\AA)} \ZZ \PPi_{\im(\AA)} - \PPi_{\im(\AA)}}_2 \le \eps$.
So, 
\[
1-\eps \le \lambda_{\min} (\AA \PPi_{\im(\AA)} \ZZ \PPi_{\im(\AA)}) \le 
\lambda_{\max} (\AA \PPi_{\im(\AA)} \ZZ \PPi_{\im(\AA)}) \le 1+ \eps.
\]
Since matrices $\AA \PPi _{\im(\AA)}\ZZ \PPi_{\im(\AA)}$ and $\AA^{1/2} \PPi_{\im(\AA)} \ZZ \PPi_{\im(\AA)} \AA^{1/2}$ have the same eigenvalues, we have $(1-\eps) \AA^{\dagger} \pleq \PPi_{\im(\AA)} \ZZ \PPi_{\im(\AA)} \pleq (1+\eps) \AA^{\dagger} $.
\end{proof}

\subsection{Missing Proofs in Section \ref{sect:alg_overview_new}}

\label{sect:proof_algo_overview}

\begin{proof}[Proof of Lemma \ref{lem:combine_solvers}]
Since $\bb \in \im(\Lup_1)$, we know $\bb_F \in \im(\Lup_1[F,F])$ and $\bb_F - \Lup_1[F,C]\Tilde{\xx}_C \in \im(\Lup_1[F,F])$. 
We can apply the solver \upFsolve \  to these two vectors.
By the statement assumption,
we an write $ \upFsolve (\bb_F) = \Lup_1[F,F]^{\dagger} \bb_F + \yy$, where $\partial_2[F,:]^{\top} \yy = {\bf 0}$. Then, 
\[
\hh = \bb_C - \Lup_1[C,F] \cdot \upFsolve(\bb_F)
= \bb_C - \Lup_1[C,F] \cdot \Lup_1[F,F]^{\dagger} \bb_F.
\]
We first show that $\hh \in \im(\schur[\Lup_1]_C)$ so that we can apply the solver $\schursolve$ to $\hh$ and obtain a vector $\Tilde{\xx}_C$ satisfying $\norm{\schur[\Lup_1]_C \tilde{\xx}_C - \hh}_2 \le \delta \norm{\hh}_2$. 
Since $\bb \in \im(\Lup_1)$, there exists $\xx = \begin{pmatrix}
\xx_F \\
\xx_C
\end{pmatrix}$ such that 
\begin{align}
\begin{split}
\begin{pmatrix}
\bb_F \\
\bb_C
\end{pmatrix} & = 
 \begin{pmatrix}
    \II & \\
    \Lup_1[C, F] \Lup_1[F, F]^{\dagger} & \II 
  \end{pmatrix} 
\begin{pmatrix}
\Lup_1[F, F] & \\
& \schur[\Lup_1]_C
\end{pmatrix}
\begin{pmatrix}
\II & \Lup_1[F, F]^{\dagger} \Lup_1[F, C] \\
& \II
\end{pmatrix}   \begin{pmatrix}
\xx_F \\
\xx_C
\end{pmatrix} \\
& = \begin{pmatrix}
    \II & \\
    \Lup_1[C, F] \Lup_1[F, F]^{\dagger} & \II 
  \end{pmatrix} \begin{pmatrix}
 \Lup_1[F,F] & \\
 & \schur[\Lup_1]_C
  \end{pmatrix} \begin{pmatrix}
\xx_F + \Lup_1[F,F]^{\dagger} \Lup_1[F,C] \xx_C \\
\xx_C
  \end{pmatrix} \\
& = \begin{pmatrix}
    \II & \\
    \Lup_1[C, F] \Lup_1[F, F]^{\dagger} & \II 
  \end{pmatrix} \begin{pmatrix}
\Lup_1[F,F] \xx_F + \Lup_1[F,C] \xx_C \\
\schur[\Lup_1]_C \xx_C
  \end{pmatrix} \\
& = \begin{pmatrix}
\Lup_1[F,F] \xx_F + \Lup_1[F,C] \xx_C \\
\Lup_1[C,F] \xx_F + \Lup_1[C, F] \Lup_1[F, F]^{\dagger} \Lup_1[F,C]\xx_C 
+ \schur[\Lup_1]_C \xx_C
\end{pmatrix}.
\end{split}
\label{eqn:lx}
\end{align}
Here, the third equality holds since 
\[
\Lup_1[F,F] \Lup_1[F,F]^{\dagger} \Lup_1[F,C] = \PPi_{\im(\partial_2[F,:])} \Lup_1[F,C] = \Lup_1[F,C],
\]
and the fourth equality holds similarly by using symmetry.
Thus, 
\begin{align*}
\bb_C & = \Lup_1[C,F] \xx_F + \Lup_1[C, F] \Lup_1[F, F]^{\dagger} \left( \bb_F - \Lup_1[F,F] \xx_F
\right)
+ \schur[\Lup_1]_C \xx_C \\
& = \Lup_1[C,F] \xx_F + \Lup_1[C,F] \Lup_1[F,F]^{\dagger} \bb_F
- \Lup_1[C,F] \xx_F + \schur[\Lup_1]_C \xx_C \\
& = \Lup_1[C,F] \Lup_1[F,F]^{\dagger} \bb_F
+ \schur[\Lup_1]_C \xx_C.
\end{align*}
That is, 
$\hh  \in \im(\schur[\Lup_1]_C)$.

Next we look at $\Lup_1 \tilde{\xx}$.
Let $\ddelta = \schur[\Lup_1]_C \tilde{\xx}_C - \hh$.
We replace $\xx_F$ and $\xx_C$ in Equation \eqref{eqn:lx} with $\tilde{\xx}_F$ and $\tilde{\xx}_C$: 
\begin{align*}
\Lup_1 \tilde{\xx}
& = \begin{pmatrix}
    \II & \\
    \Lup_1[C, F] \Lup_1[F, F]^{\dagger} & \II 
  \end{pmatrix} \begin{pmatrix}
\Lup_1[F,F] \tilde{\xx}_F + \Lup_1[F,C] \tilde{\xx}_C \\
\schur[\Lup_1]_C \tilde{\xx}_C
  \end{pmatrix}
  \\
& = \begin{pmatrix}
\bb_F  \\
\Lup_1[C,F]\Lup_1[F,F]^{\dagger} \bb_F + \schur[\Lup_1]_C \tilde{\xx}_C
\end{pmatrix} \\
& = \begin{pmatrix}
\bb_F \\
\bb_C + \ddelta
\end{pmatrix}.
\end{align*}
Then,
\begin{multline*}
\norm{\Lup_1 \tilde{\xx} - \bb}_2 
= \norm{\ddelta}_2 
\le \delta \norm{\hh}_2 
\le \delta \left( \norm{\bb_C}_2 + \norm{\Lup_1[C,F] \Lup_1[F,F]^{\dagger}}_2 \norm{\bb_F}_2
\right) \\
\le \delta \left( 1 + \norm{\Lup_1[C,F] \Lup_1[F,F]^{\dagger}}_2
\right) \norm{\bb}_2
\le \eps \norm{\bb}_2,
\end{multline*}
where the last inequality is by our setting of $\delta$.

We compute $\tilde{\xx}$ by two calls of $\upFsolve$ and one call of $\schursolve$ and $O(1)$ matrix-vector multiplications and vector-vector additions. Thus, 
the total runtime is $O(t_1(m_F) + t_2(m_C) + 
\text{the number of nonzeros in } \Lup_1)$.
\end{proof}